\documentclass[lettersize,journal，comsoc]{IEEEtran}
\usepackage{amsmath,amsfonts}
\usepackage{algorithmic}
\usepackage{algorithm}
\usepackage{array}
\usepackage[caption=false,font=normalsize,labelfont=sf,textfont=sf]{subfig}
\usepackage{textcomp}
\usepackage{stfloats}
\usepackage{url}
\usepackage{multirow}
\usepackage{booktabs}
\usepackage{verbatim}
\usepackage{graphicx}
\usepackage{cite}
\usepackage{color}
\usepackage{amsthm,amssymb}
\newtheorem{theorem}{Theorem}

\newtheorem{assumption}{Assumption}
\newtheorem{proposition}{Proposition}
\begin{document}

\title{Adversarial Water-Filling: Theory, Algorithms, and a Domain-Specific Wireless Foundation Model}

\author{Xindi Tong, Chee Wei Tan, H. Vincent Poor\thanks{
% This work was supported in part by the Singapore Ministry of Education Academic Research Fund and MOE-T2EP20224-0009. The work of H. V. Poor was supported in part by an Innovation Grant from Princeton NextG.
Xindi Tong and Chee Wei Tan are with the College of Computing and Data Science (CCDS), Nanyang Technological University, Singapore. 
H. Vincent Poor is with the Department of Electrical and Computer Engineering, Princeton University, Princeton, NJ USA.
E-mail: to0002di@e.ntu.edu.sg,  cheewei.tan@ntu.edu.sg, poor@princeton.edu.
\textit{(Corresponding author: Chee Wei Tan.)}
}
% \thanks{$^{*}$Corresponding author: Xindi Tong.}}
% Manuscript received April 19, 2005; revised August 26, 2015.}
}

\markboth{Journal of \LaTeX\ Class Files,~Vol.~14, No.~8, August~2021}%
{Shell \MakeLowercase{\textit{et al.}}: A Sample Article Using IEEEtran.cls for IEEE Journals}


\maketitle

\begin{abstract}

Competitive resource allocation problems over frequency and space can be formulated as a minimax interaction between transmit power and worst-case interference. This formulation naturally arises in multi-operator low Earth orbit (LEO) satellite spectrum sharing, where transmissions from competing constellations interfere in real-time. \color{black}
Under Gaussian channels, the corresponding power-allocation problem admits a convex--concave formulation with a unique saddle point. Discrete constellations yield generally nonconvex mercury/water-filling formulations.
\color{black}
\color{black}
In this paper we propose the adversarial water-filling (AWF) problem with corresponding theory and algorithms for these settings. In addition, we develop a domain-specific
\color{black}
wireless foundation model for AWF to learn the AWF search
dynamics. The architecture incorporates permutation-invariant
channel representations, a constraint-aware graph neural
network (GNN) with sparse message passing, and global latent
variables capturing the low-dimensional water level implied by
the AWF optimality. Through learned projected extragradient
iterations, the model approximates stationary solutions of the
constrained minimax problem arising under mercury/water-filling.
\color{black}
We further establish projected-stationarity/Karush-Kuhn-Tucker consistency and
conditional local convergence of the learned AWF dynamics under
local regularity and stability conditions.
Experiments demonstrate empirical generalization across unseen
problem sizes, constraint structures, discrete constellations, and
channel-weighted objectives, while achieving a median speedup exceeding one order of magnitude over Mirror-Prox on matched instances at comparable first-order solution quality. \color{black} The related code can be found at https://github.com/convexsoft/AWF. 

\end{abstract}


\section{Introduction}

The sixth generation (6G) mobile communication network is expected to be AI-native, where intelligence is embedded throughout network design, operation, and service provisioning~\cite{wang6g2023}. This requires scalable and autonomous resource adaptation as a key system principle~\cite{liu2019dynamic}.
Within this paradigm, Open Radio Access Network (Open RAN) has emerged as a transformative architecture enabling multi-vendor interoperability through standardized open interfaces~\cite{HABLER2025104090}. 
Such architectural flexibility becomes particularly important as next-generation networks integrate high-altitude platform stations and non-terrestrial networks (NTNs) to support large-scale low Earth orbit (LEO) constellations~\cite{OpenRANxApps,OpenRANINFOCOM,chen2025openranetneuralizedspectrumaccess}. 
Recent LEO direct-to-device proceedings involving SpaceX Starlink and Omnispace have raised coexistence concerns, highlighting competitive interference and spectrum-coordination challenges under mobility, beam dynamics, directional links, and overlapping spectrum use~\cite{fcc2024spacex}.

\begin{figure}
    \centering
    \includegraphics[width=0.95\linewidth]{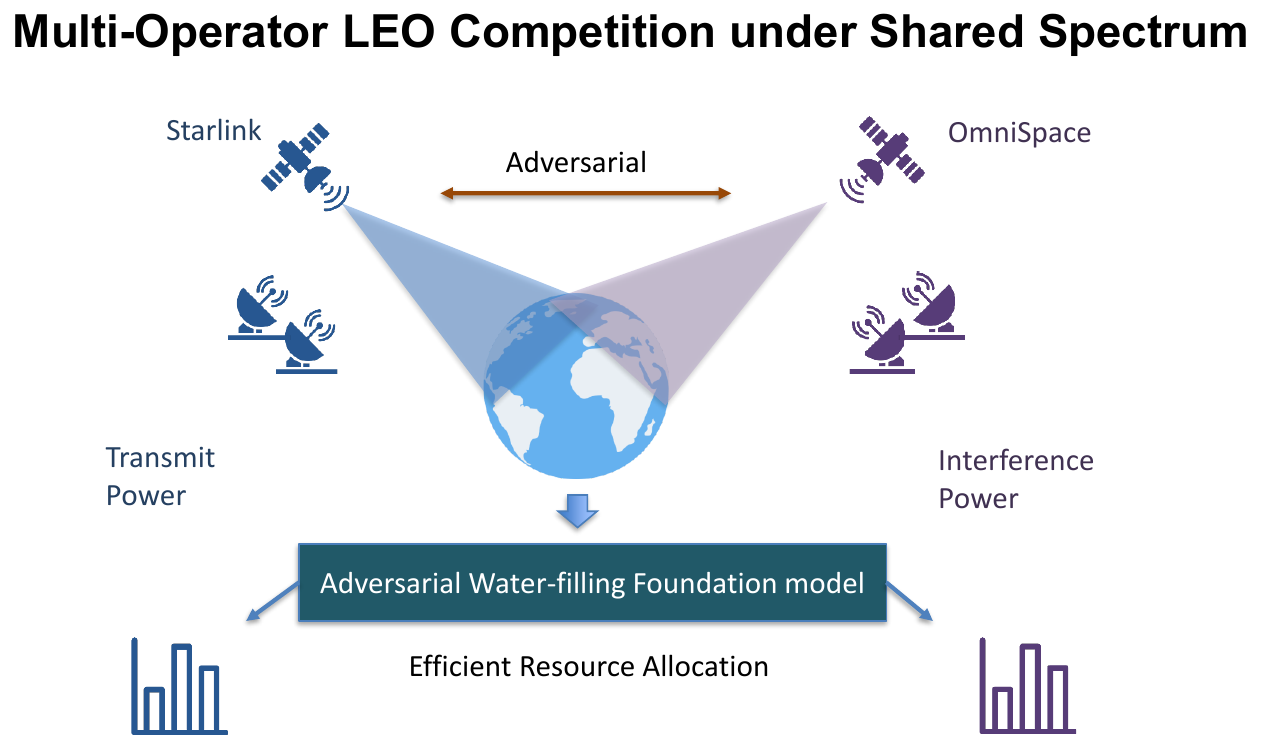}
\caption{Competition between Starlink and Omnispace \textcolor{black}{motivates} a minimax resource-allocation problem, for which the AWF foundation model provides an efficient solution.}
    \vspace{-3mm}
    \label{fig:example1}
\end{figure}

Unlike traditional cellular systems with relatively fixed base-station deployments, NTN coexistence involves space--time varying interference. A representative example is the emerging coexistence between Starlink direct-to-cell services and Omnispace-like mobile-satellite systems under shared or adjacent spectrum and regulatory constraints~\cite{fcc2024spacex}. Highly directional beams, overlapping footprints, heterogeneous antenna gains, and terrestrial--satellite spectrum reuse create location-dependent interference, while satellite motion, beam steering, short visibility windows, and handovers continuously alter the dominant interferers and coupling strengths~\cite{li2024codedwaterfillingmultiuserinterference,Hraishawi2023}. Consequently, one operator's transmission may become uncertain and rapidly varying interference to another, motivating adversarial resource allocation under worst-case space--time interference (cf. Fig.~\ref{fig:example1}). Existing studies identify power flux density limits, equivalent isotropically radiated power masks, orbital geometry, beam directivity, frequency reuse, and payload flexibility as key factors governing NTN interference footprints and coexistence performance~\cite{Hraishawi2023,giambene2018satellite,kassing2020exploring,oranspectrumsharing}. These characteristics motivate a snapshot-level optimization framework that captures spatial coupling and competitive interference across heterogeneous operating conditions. In this framework, resource allocation is performed within individual snapshots, with channel parameters, interference relationships, and spatial constraints specified for each snapshot and allowed to vary across operating conditions.

From an information-theoretic perspective, water-filling provides the fundamental mechanism for power allocation over parallel Gaussian channels~\cite{gallager1968information}. For practical discrete constellations, mercury/water-filling extends this principle through modulation-dependent ``mercury levels'' derived from the I--MMSE relationship~\cite{lozano2006optimum}. While Gaussian water-filling admits well-characterized solutions under concave utilities, mercury/water-filling under adversarial interference can exhibit more intricate utility curvature and nonconvex--nonconcave structure. This motivates a minimax formulation for competitive NTN spectrum sharing. Water-filling also admits a proximal interpretation involving a small number of dual variables~\cite{palomar2005practical}. Under spatial constraints and multi-operator competition~\cite{xingNewviewpoint2020,zhengWirelessMaxMin2016}, these dual variables act as water levels coordinating high-dimensional channel allocations. This structure connects water-filling to primal--dual methods such as the primal--dual hybrid gradient (PDHG) algorithm~\cite{chambolle2011first,goldstein2013adaptive} and distributed multi-agent optimization~\cite{Yang2010}. However, conventional model-based methods generally solve each network instance separately, requiring repeated optimization as channel dimensions, constraints, and modulation distributions change.

Learning-based approaches have been explored to accelerate wireless resource allocation~\cite{Sun_2018}, but many remain task-specific and dimension-dependent. Recent advances in wireless-oriented large language models have highlighted the potential of integrating domain knowledge, physical constraints, and external optimization tools into wireless decision-making~\cite{wireless10582827}. These developments motivate generalizable learning-based optimization frameworks that leverage shared problem structures to reuse optimization strategies across related resource-allocation problems.

In this paper, we introduce adversarial water-filling (AWF) as a unified framework for competitive wireless resource allocation over frequency and space, covering adversarial Gaussian water-filling and mercury/water-filling in competitive NTN spectrum sharing. 
AWF naturally supports a foundation-model approach through channel permutation invariance, sparse constraint-induced interactions, and global water-level coordination via dual variables. 
Inspired by classical Gaussian water-filling~\cite{palomar2005practical} and primal--dual/proximal methods~\cite{chambolle2011first,goldstein2013adaptive}, we develop a domain-specific wireless foundation model with permutation-invariant channel representations, constraint-graph propagation, and learned primal--dual dynamics. 
\color{black}
The proposed model is reused without instance-specific retraining across channel dimensions, constraint patterns, modulation formats, and channel-weighted objective variants, demonstrating structured zero-shot generalization within the AWF family.

The technical distinction lies in combining the AWF minimax formulation with a reusable architecture that explicitly reflects its optimization structure. Such reuse is particularly relevant to NTN operating snapshots with varying channel, interference, and constraint configurations. A preliminary conference version appeared in~\cite{tong2025}. It focused on Gaussian AWF in Open RAN, the proximal/PDHG interpretation, and finite-step water-level search. This journal version substantially extends that work to spatially constrained and mercury/water-filling AWF, introducing a domain-specific foundation model, Karush-Kuhn-Tucker (KKT) and conditional local-stability analysis, and extensive generalization, ablation, and NTN-inspired sensitivity experiments. The main contributions of this paper are summarized as follows:
\begin{itemize}
\color{black}
\item We formulate adversarial water-filling as a minimax problem for competitive wireless resource allocation under spatial constraints, with multi-operator NTN spectrum sharing as a motivating operating regime.
\color{black}
\item We unify Gaussian water-filling and mercury/water-filling within a common minimax formulation and identify both convex--concave regimes and more challenging nonconvex regimes.

\color{black}
\item We develop a domain-specific wireless foundation model
for AWF with permutation-invariant encoding, constraint-aware
graph representations, and learned primal--dual dynamics,
supporting reuse across dimensions, constraints, modulation
formats, and channel-weighted objective variants without
retraining.
\color{black}

\item We establish
\color{black}
projected-stationarity/KKT consistency and conditional local
stability of the learned AWF dynamics under explicit
regularity and contraction conditions, clarifying the scope of
standard first-order optimality and local stability arguments
in the proposed model.
\color{black}

\end{itemize}

\section{Related Work}
\label{sec:related}

% \subsection{Open RAN and Multi-Operator NTN Coexistence}

% The transition toward 6G has accelerated the adoption of Open RAN architectures that emphasize modularity, multi-vendor interoperability, and AI-driven control~\cite{oransurveytutorial2023,oranspectrumsharing,OpenRANxApps,OpenRANINFOCOM}. 
% In parallel, NTNs, particularly LEO constellations such as Starlink and Omnispace, are becoming integral to global connectivity~\cite{habler2023adversarialmachinelearningthreat}. 
% Their directional beams, rapid orbital motion, and dynamic payload configurations produce highly spatial and time-varying interference patterns. 
% When multiple operators share or operate in adjacent frequency bands, competitive interference becomes unavoidable. Existing studies analyze multi-operator coexistence under regulatory constraints such as power flux density (PFD) limits and equivalent isotropically radiated power (EIRP) masks, showing how orbital geometry and beam directivity shape interference footprints~\cite{Hraishawi2023}. 
% System-level analyses further demonstrate that narrow beams, frequency reuse, and payload flexibility significantly influence coexistence performance~\cite{giambene2018satellite,kodheli2020satellite}, while rapidly evolving constellation topologies introduce time-varying spatial coupling~\cite{kassing2020exploring}. 
% These characteristics indicate that power allocation in multi-operator NTN systems is inherently spatially coupled and competitive, motivating optimization frameworks that explicitly model interference interactions.

\subsection{Water-Filling and Minimax Games}

Water-filling is a fundamental principle for information-theoretic power allocation, originating from optimal resource allocation over parallel Gaussian channels~\cite{cover1991elements}. 
It has been widely applied to SISO, OFDM, and MIMO systems~\cite{tse2005fundamentals}. 
For practical discrete constellations, mercury/water-filling introduces modulation-dependent corrections through the I--MMSE relationship, thereby extending classical results beyond Gaussian signaling~\cite{lozano2006optimum}. 
Iterative water-filling algorithms have also been extensively studied for Gaussian multiuser channels and are known to converge to optimal or equilibrium power allocations under suitable conditions~\cite{IWF2004}. 
This naturally leads to a game-theoretic interpretation: in multiuser interference channels, each transmitter performs water-filling against the interference generated by other users, while the coupled sum-rate maximization problem is generally nonconvex~\cite{yu2006dual}. 

When interference is no longer a fixed impairment but an uncertain and potentially adversarial action, power allocation can be formulated as a two-player minimax problem.  
Such minimax formulations have been extensively studied in convex optimization and game theory, with classical results establishing the existence of saddle points and efficient methods for computing equilibria~\cite{ghosh2003minimax}. 
However, most existing water-filling formulations assume known channel statistics and predictable interference levels, which limits their applicability to competitive NTN environments where interference may be uncertain, nonstationary, and operator dependent.

First-order methods, including PDHG and related proximal algorithms, provide practical tools for solving large-scale instances of such problems~\cite{chambolle2011first,goldstein2013adaptive}. 
Worst-case and max--min formulations have appeared in anti-jamming communications, where transmit strategies are designed to maintain throughput under uncertain or malicious interference~\cite{Garnaevmaxmin2025}. 
Related multi-access game models further examine interactions between selfish users and adaptive jammers under Nash or Stackelberg formulations~\cite{garnaev2022antijamming}.  Existing game-theoretic and water-filling methods are often developed for fixed dimensions, prescribed uncertainty sets, or Gaussian signaling. The AWF formulation considered here combines worst-case interference with spatially coupled constraints and both Gaussian and mercury/water-filling utilities. Building on this structure, we develop a domain-specific wireless foundation model that learns reusable primal–dual water-level dynamics across varying problem configurations. The resulting approach retains the structure of model-based first-order optimization while enabling reuse across related AWF settings.
\color{black}

% In wireless networks, power allocation and interference management have also been studied through game-theoretic frameworks.  Distributed power control has been analyzed using Nash equilibrium and variational inequality theory, while network utility maximization and dual decomposition connect rate and power allocation to convex optimization models~\cite{chiang2017power,yu2002distributed,tan2015wireless}. 
\subsection{Learning-Based Spectrum Management}

Learning has been widely explored to accelerate wireless resource allocation under complex and dynamic network conditions. 
Early work used neural networks to approximate optimization solution mappings for fast inference~\cite{Sun_2018}, later extending to dynamic settings via continual and bilevel optimization~\cite{Sun_2022}. 
Reinforcement learning and multi-agent methods have been applied to distributed power control and spectrum sharing~\cite{Nasir_2019}, while graph neural networks can exploit wireless channel graphs to capture interference coupling and improve scalability in resource allocation~\cite{shen2021gnnrrm}. 
Model-based deep learning can improve interpretability and robustness by embedding iterative optimization principles into neural architectures~\cite{yang2024knowledge}. However, unpredictable NTN coexistence introduces additional challenges beyond conventional learning-based spectrum management. 
The number of visible satellites, active beams, dominant interferers, constraint patterns, and modulation formats may change across time, space, and operators. 
As a result, models trained for a fixed topology or a fixed channel dimension may fail to generalize under distribution shifts caused by satellite mobility, beam steering, handovers, and adversarial interference. 
Permutation-invariant architectures, such as Deep Sets, Set Transformer, and Perceiver, provide useful tools for size-generalized modeling of unordered channel sets~\cite{deepsets2017,lee2019settransformerframeworkattentionbased,jaegle2021perceivergeneralperceptioniterative}, while graph-based representations can encode spatial coupling induced by interference and regulatory constraints.
\color{black}
Graph Transformers provide a complementary approach by
combining graph structure with attention mechanisms and have
shown strong capability in graph representation learning
\cite{ying2021}. More recently, graph-based
foundation models have also been investigated for wireless
resource allocation
\cite{sheng2026graphfoundationmodelwireless}. These developments motivate
the comparison between message-passing and attention-based graph
representations for scalable wireless optimization.
\color{black}

\color{black}
Foundation-model concepts have begun to be explored in the wireless domain. WirelessLLM \cite{wireless10582827}, for example, investigates the adaptation of large language models for wireless intelligence, illustrating the broader potential of pretrained and transferable models in wireless applications. While focusing on a different model class and task setting, it provides useful motivation for exploring reusable pretrained models in wireless optimization.
\color{black}
NVIDIA has also advanced this direction through its artificial intelligence–radio access network   initiatives (AI-RAN), while OpenRAN Gym enables large-scale data collection and learning-based experimentation on platforms for advanced wireless research (PAWR) infrastructures~\cite{nvidia2022}. 
PAWR provides programmable wireless testbeds for evaluating communication techniques and learning-based spectrum management under realistic conditions.
\textcolor{black}{These developments motivate pretrained models that capture reusable optimization structure across varying wireless configurations.}
% This motivates our constraint-varying and size-generalized AWF foundation model for competitive NTN resource allocation.

% The proposed adversarial water-filling framework integrates these ideas:
% the primal allocations are closed-form proximal steps, while the dual 
% variables are updated using a damped Newton refinement over a low-dimensional
% nonlinear system. This hybrid structure yields rapid convergence, even in
% the presence of asymmetric spatial constraints and adversarial interference,
% thereby providing a scalable solution for multi-operator coexistence in
% large NTN deployments.

\section{Adversarial Water-Filling over Frequency}

We consider a resource allocation problem with a global budget constraint, where the objective is separable across components but coupling arises through the total resource constraint. 
Water-filling characterizes this structure by equalizing marginal utilities across all active components. In conventional settings with fixed receiver noise, this leads to optimal power-allocation strategies for parallel Gaussian channels \cite{boyd2004convex}. 
We extend this framework to an AWF problem in which the receiver faces variable interference subject to a global interference budget. 
The resulting formulation naturally leads to a minimax resource-allocation problem between transmit power and adversarial interference, reminiscent of worst-case power control and spectrum allocation in wireless networks over frequency \cite{ghosh2003minimax,spectrumTan2011}.

Consider $m$ channels with transmit power 
$\mathbf{p}\in\mathbb{R}^m_{\ge0}$ and interference power 
$\mathbf{n}\in\mathbb{R}^m_{\ge0}$.  
Channel $i$ has channel gain $\beta_i>0$ and background noise 
$\sigma_i>0$.  
Both players satisfy global budgets 
$\mathbf{1}^\mathsf{T}\mathbf{p}=P$ and 
$\mathbf{1}^\mathsf{T}\mathbf{n}=N$. 
The AWF problem over frequency is formulated as
\begin{equation}
\max_{\mathbf{p}\ge 0,\ \mathbf{1}^\mathsf{T}\mathbf{p}=P}\ 
\min_{\mathbf{n}\ge 0,\ \mathbf{1}^\mathsf{T}\mathbf{n}=N}\ 
f(\mathbf{p},\mathbf{n}),
\label{eq:spectrum-awf-general}
\end{equation}
where $f(\mathbf{p},\mathbf{n})$ depends on the modulation model.

\subsection{Gaussian Water-filling}

For Gaussian channels, the achievable sum capacity (in nats per channel use) is
\begin{equation}
f(\mathbf p,\mathbf n) = \sum_{i=1}^m \log\!\left(1+\frac{\beta_i p_i}{\sigma_i+n_i}\right).
\end{equation}

For fixed $\mathbf n$, $f(\mathbf p,\mathbf n)$ is 
concave in $\mathbf p$ over the feasible simplex, while for
fixed $\mathbf p$, $f(\mathbf p,\mathbf n)$ is convex in
$\mathbf n$.

\color{black}
The convexity with respect to $\mathbf n$ is strict on the active support
of $\mathbf p$, but is not uniformly strong over the entire simplex because
the curvature vanishes when $p_i=0$. Under $\beta_i>0$, $\sigma_i>0$,
$P>0$, and a nonempty feasible set, the Gaussian AWF nevertheless admits
a unique saddle point. For any fixed feasible $\mathbf p$, its support
is nonempty. Allocating interference to a coordinate with $p_i=0$ does not
affect the objective, whereas reallocating that interference to an active
coordinate strictly decreases it. Hence, every interference best response
places no power outside the active support of $\mathbf p$. Restricted to
this support, the interference subproblem is strictly convex and therefore
has a unique minimizer. For any fixed feasible $\mathbf n$, the objective
is strictly concave in $\mathbf p$, so the transmit-power best response is
also unique. Consequently, if two distinct saddle pairs existed, the inequalities would imply multiple best responses to the same opposing
allocation, contradicting the uniqueness of the corresponding best
responses. Thus, the Gaussian AWF admits a unique saddle pair.
\color{black}
The resulting problem is formulated as follows:
\begin{equation}
\max_{\mathbf p\ge 0} \min_{\mathbf n\ge 0}
\sum_{i=1}^m \log\left(1+\frac{\beta_i p_i}{\sigma_i+n_i}\right),
\label{eq:min-max-problem}
\end{equation}
subject to $\mathbf{1}^\mathsf{T} \mathbf p=P$ and $\mathbf{1}^\mathsf{T} \mathbf n=N$.

\noindent\textit{Optimal Power Allocation.}
For a fixed interference vector $\mathbf n$, the maximization over $\mathbf p$ reduces to the water-filling problem. The optimal transmit power allocation takes the form
\begin{equation}
p_i^\star = \max\left(\nu-\frac{\sigma_i+n_i^\star}{\beta_i},0\right), \quad \forall i
\label{eq:p_optimal}
\end{equation}
where $\nu$ is the water level determined by the power constraint.

Conversely, for a fixed transmit power allocation $p$, the adversarial interference allocation minimizing the capacity is obtained via the corresponding Lagrangian formulation, leading to the closed-form expression
\begin{equation}
n_i^\star =
\max\left(
-\frac{\beta_i p_i^\star}{2}
+\frac{1}{2}\sqrt{(\beta_i p_i^\star)^2-4\frac{\beta_i p_i^\star}{\mu}}
-\sigma_i,0
\right), \quad \forall i
\label{eq:n_op}
\end{equation}
where $\mu$ is the dual variable associated with the total interference power constraint $\mathbf{1}^\mathsf{T}\mathbf n=N$ \cite{ghosh2003minimax}. The minimax formulation of the water-filling problem reveals a strong alignment between active transmit and interference channels, as formalized in the following result.

\begin{theorem}[Adversarial Water-Filling over Frequency]
Consider the minimax problem in \eqref{eq:min-max-problem}. 
Let $(\mathbf{p}^\star,\mathbf{n}^\star)$ denote an optimal solution.
\label{thm:awf_frequency}
\begin{itemize}
\item \emph{Inactive channels:} For any $i\in\{1,\dots,m\}$, 
if $p_i^\star=0$, then the corresponding interference allocation 
satisfies $n_i^\star=0$.

\item \emph{Active channels:} For any $i\in\{1,\dots,m\}$ with 
$p_i^\star>0$ and $n_i^\star>0$, the interference dual variable 
$\mu$ satisfies
$\mu
=
\frac{1}{\beta_i\nu}
-
\frac{1}{\sigma_i+n_i^\star}$.
\end{itemize}
\end{theorem}

\begin{proof}
If $p_i^\star=0$, substituting into the closed-form expression 
\eqref{eq:n_op} gives $n_i^\star=0$ since $\sigma_i>0$. 
Hence, channels inactive for transmission also receive zero 
interference allocation \cite{tong2025}.

For channels with $p_i^\star>0$, consider the interference 
minimization subproblem with fixed $\mathbf{p}^\star$. 
The Lagrangian is
\begin{equation}
\mathcal{L}(\mathbf{n},\mu,\mathbf{c})
=
\sum_{i=1}^m 
\log\!\left(
1+\frac{\beta_i p_i^\star}{\sigma_i+n_i}
\right)
+\mu\!\left(N-\mathbf{1}^\mathsf{T}\mathbf{n}\right)
-\mathbf{c}^\mathsf{T}\mathbf{n},
\end{equation}
where $\mu$ and $\mathbf{c}\ge0$ are the dual variables. 
For any channel $i$ with $n_i^\star>0$, complementary slackness 
gives $c_i=0$, and the stationarity condition yields
\begin{equation}
-\frac{\beta_i p_i^\star}
{(\sigma_i+n_i^\star)
(\sigma_i+n_i^\star+\beta_i p_i^\star)}
-\mu
=0.
\end{equation}

Substituting the transmit water-filling solution 
\eqref{eq:p_optimal} into the above equation gives
$\mu
=
\frac{1}{\beta_i\nu}
-
\frac{1}{\sigma_i+n_i^\star}$,
for any $i\in\{1,\dots,m\}$ with 
$p_i^\star>0$.
\end{proof}

The above relation shows that the adversarial interference 
allocation induces a water level $\mu$ coupled with the 
transmit-side water level $\nu$. The AWF problem therefore 
yields dual water levels for transmit power and interference 
power, as illustrated in Fig.~\ref{fig:twofigs}.
\begin{figure}[t]
\centering

\begin{minipage}{0.49\linewidth}
\centering
\includegraphics[width=\linewidth]{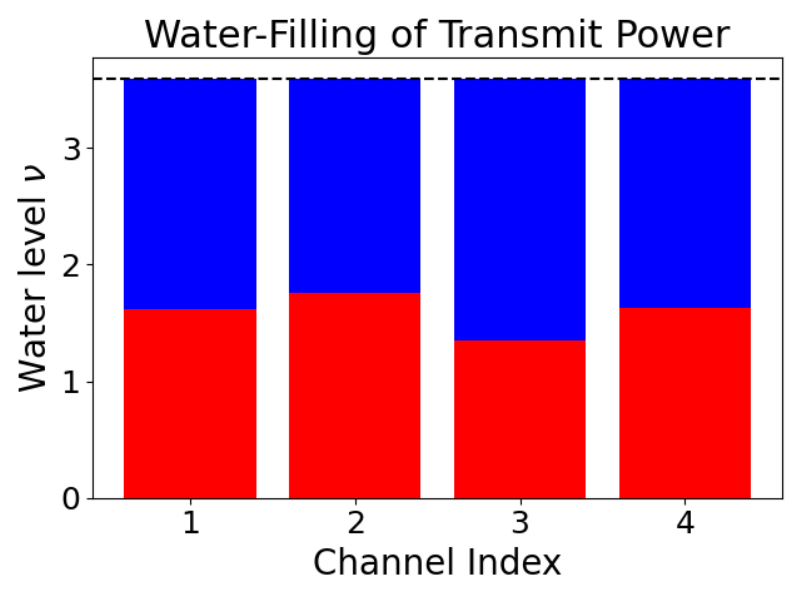}
\end{minipage}
\hfill
\begin{minipage}{0.49\linewidth}
\centering
\includegraphics[width=\linewidth]{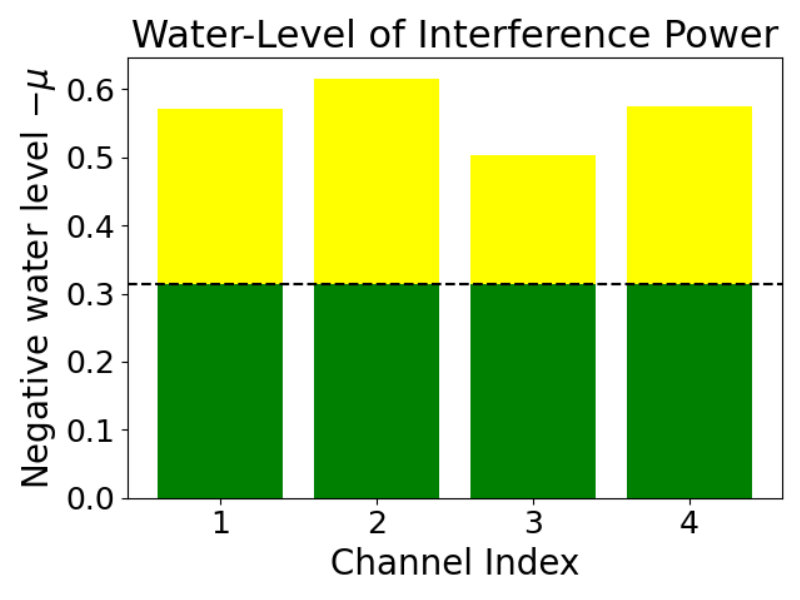}
\end{minipage}

\caption{Illustration of adversarial water-filling. 
(a) Transmit-side water-filling for $\mathbf p^\star$, where the dashed line denotes the water level $\nu$, red bars denote effective base levels $(\boldsymbol{\sigma}+ \mathbf n^\star)/\boldsymbol \beta$, and blue segments denote powers $\mathbf p^\star$. 
(b) Interference allocation for $\mathbf n^\star$, where the dashed line denotes the interference water level $-\mu$, the total bar height is $1/(\boldsymbol \sigma+\mathbf n^\star)$, and yellow segments denote $1/(\boldsymbol\beta\nu)$.}
\vspace{-3mm}
\label{fig:twofigs}
\end{figure}

\subsection{Mercury/water-filling}

We keep the same variables and constraints $(\mathbf p,\mathbf n,$ $P,N, \beta_i,\sigma_i)$, but replace the logarithmic utility by a mutual-information objective induced by a discrete input constellation:
\begin{equation}
f(\mathbf  p,\mathbf n)=\sum_{i=1}^{m} I_i\!\left(\frac{\beta_i p_i}{\sigma_i+n_i}\right),
\qquad s.t. \quad
\gamma_i=\frac{\beta_i p_i}{\sigma_i+n_i}.
\end{equation}
Here $I_i(\gamma)$ denotes the mutual information (in nats/use) of channel $i$ as a function of the SINR $\gamma$. For Gaussian inputs, $I_i(\gamma)=\log(1+\gamma)$ and the model reduces to the Gaussian-channel case. The minimax game becomes 
\begin{equation}
\max_{\mathbf p\ge 0,\ \mathbf{1}^\mathsf{T}p=P}\ \min_{\mathbf n\ge 0,\ \mathbf{1}^\mathsf{T}n=N}\ 
\sum_{i=1}^{m} I_i\!\left(\gamma_i\right).
\label{eq:min-max-mercury}
\end{equation}

\color{black}
Theorem~\ref{thm:awf_frequency} can be naturally extended at the level of first-order stationarity to the mercury/water-filling problem over frequency. For fixed interference power $n$, the transmit-power stationarity condition for any active channel $i\in\{1,\dots,m\}$ with $p_i^\star>0$ is
\color{black}
\begin{equation}
\frac{\beta_i}{\sigma_i+n_i^\star}\,
I_i'\!\left(\gamma_i^\star\right)
=
\nu,
\label{eq:p_mercury}
\end{equation}
\color{black}
where $\nu$ enforces $\mathbf{1}^\mathsf{T}\mathbf p=P$. The interference-power stationarity condition for any $i\in\{1,\dots,m\}$ with $n_i^\star>0$ is
\color{black}
\begin{equation}
-\frac{\beta_i p_i^\star}{(\sigma_i+n_i^\star)^2}\,
I_i'\!\left(\gamma_i^\star\right)
=
\mu,
\label{eq:n_mercury}
\end{equation}
where $\mu$ enforces $\mathbf{1}^\mathsf{T}\mathbf n=N$. By the I--MMSE identity\footnote{The I--MMSE identity establishes a fundamental
relationship between the derivative of mutual information
with respect to the signal-to-noise ratio (SNR) and the
minimum mean-squared error (MMSE) in point-to-point
Gaussian channels.}\cite{guo2005mutual},
\begin{equation}
I_i'(\gamma)=\mathrm{mmse}_i(\gamma), \quad \forall i
\label{eq:mmse}
\end{equation}
\eqref{eq:p_mercury}--\eqref{eq:n_mercury} define a dual water-filling system, with transmit-side and interference-side water levels $\nu$ and $\mu$, respectively.

\section{Adversarial Water-Filling over Space}
\label{sec:spatial-awf}

We develop an AWF framework over space by augmenting
\eqref{eq:spectrum-awf-general} with linear power-shaping
constraints. Consider the same power-allocation vectors
$\mathbf{p}$ and $\mathbf{n}$, budgets $P$ and $N$, and
channel parameters $\beta_i$ and $\sigma_i$. We impose the
linear constraint
$\mathbf{A}\mathbf{p}\leq\hat{\mathbf{p}}$, where
$\mathbf{A}\in\mathbb{R}_{\geq0}^{S\times m}$ is a sparse
nonnegative matrix encoding power limits across channels and
$\hat{\mathbf{p}}\in\mathbb{R}^S$ denotes the corresponding
constraint-threshold vector. The AWF problem over space is
\begin{equation}
\max_{\substack{
\mathbf{p}\geq0,\;
\mathbf{1}^{\mathsf T}\mathbf{p}=P,\\
\mathbf{A}\mathbf{p}\leq\hat{\mathbf{p}}
}}
\;
\min_{\mathbf{n}\geq0,\;
\mathbf{1}^{\mathsf T}\mathbf{n}=N}
f(\mathbf{p},\mathbf{n}),
\label{eq:spatial-awf-general}
\end{equation}
where $f(\mathbf{p},\mathbf{n})$ depends on the modulation model.

%--------------------------------------------------
\subsection{Gaussian Water-filling}
%--------------------------------------------------

Under Gaussian water-filling, the achievable sum rate is
\begin{equation}
f(\mathbf p,\mathbf n)
=
\sum_{i=1}^m 
\log\!\left(1+\frac{\beta_i p_i}{\sigma_i+n_i}\right).
\label{eq:spatial-gaussian}
\end{equation}

\color{black}
For fixed $\mathbf n$, the objective is strictly concave in
$\mathbf p$ over the feasible power polytope. For fixed
$\mathbf p$, it is convex in $\mathbf n$ and strictly convex over the active interference coordinates. 
Consequently, for Gaussian inputs, AWF admits a
convex--concave max--min formulation:
\color{black}
\begin{equation}
\max_{\mathbf p\ge 0} \;\min_{\mathbf n\ge 0}\; 
\sum_{i=1}^m \log\!\left(1+\frac{\beta_i p_i}{\sigma_i+n_i}\right),
\label{eq:spatial-awf}
\end{equation}
subject to $\mathbf{1}^\mathsf{T}\mathbf p=P$, $\mathbf{1}^\mathsf{T}\mathbf n=N$, and $\mathbf A \mathbf p \le \mathbf {\hat p}$.

Problem \eqref{eq:spatial-awf} is more general than the formulation discussed in the previous section. In particular, the min--max problem in \eqref{eq:min-max-problem} can be recovered as the special case in which there are no additional spatial power constraints, i.e., $\mathbf A=0$ so that the constraint $\mathbf A \mathbf p \le \mathbf {\hat p}$ becomes inactive. Under these additional linear constraints, the classical scalar water-level characterization no longer applies directly. Instead, the resulting water-filling solution becomes multi-dimensional, with the water levels coupled through the constraint $\mathbf A \mathbf p \le \mathbf {\hat p}$.

Introduce the dual variable $\nu$ for the total transmit-power
constraint
$\mathbf{1}^\mathsf{T}\mathbf{p}=P$, and
$\boldsymbol{\theta}\ge 0$ for the linear constraint
$\mathbf{A}\mathbf{p}\le \hat{\mathbf{p}}$.
Let $\mu$ denote the dual variable associated with the total
interference-power constraint
$\mathbf{1}^\mathsf{T}\mathbf{n}=N$.
We denote by $\mathbf{b}\ge 0$ and $\mathbf{c}\ge 0$ the dual
variables associated with the nonnegativity constraints
$\mathbf{p}\ge 0$ and $\mathbf{n}\ge 0$, respectively. The Lagrangian is
\begin{equation}
\begin{aligned}
\mathcal{L}(
\mathbf{p},
\mathbf{n},
\nu,
\mu,
\boldsymbol{\theta},
\mathbf{b},
\mathbf{c})
&=
\sum_{i=1}^{m}
\log\!\left(
1+\frac{\beta_i p_i}{\sigma_i+n_i}
\right)
-\nu\!\left(
\mathbf{1}^\mathsf{T}\mathbf{p}-P
\right) \\
&\quad
+\mu\!\left(
N-\mathbf{1}^\mathsf{T}\mathbf{n}
\right)
-\boldsymbol{\theta}^\mathsf{T}
\left(
\mathbf{A}\mathbf{p}-\hat{\mathbf{p}}
\right) \\
&\quad
+\mathbf{b}^\mathsf{T}\mathbf{p}
-\mathbf{c}^\mathsf{T}\mathbf{n}.
\end{aligned}
\label{eq:L-spatial}
\end{equation}

% Since the objective induces a strongly convex–concave game and the feasible region is compact, the optimal solution is uniquely characterized by the KKT conditions.

\color{black}
The additional spatial constraints preserve convexity of the transmit-power feasible set. The objective remains strictly concave in the transmit-power variables and convex in the interference variables. As in the frequency-domain Gaussian AWF, the interference minimizer is supported only on active transmit channels, over which the interference subproblem is strictly convex. Hence the Gaussian AWF over space also admits a unique saddle point, which is characterized by the KKT conditions.
\color{black}

\begin{theorem}[Adversarial Water-filling over Space]
\label{thm:spatial-pwf}
Let $(\mathbf p^\star,\mathbf n^\star)$ denote an optimal solution of \eqref{eq:spatial-awf}. 
Then the optimal transmit power allocation satisfies
\begin{equation}
p_i^\star
=
\max\!\left(
\frac{1}{\nu^\star+(\mathbf A^\mathsf{T}\boldsymbol{\theta}^\star)_i}
-\frac{\sigma_i+n_i^\star}{\beta_i},0
\right), \quad \forall i
\end{equation}
where $\nu^\star$ and $\boldsymbol{\theta}^\star\ge0$ are the dual variables associated with the total power constraint $\mathbf{1}^\mathsf{T}\mathbf p=P$ and the linear constraint $\mathbf A\mathbf p\le \hat{\mathbf p}$, respectively.
\end{theorem}

\begin{proof}
Fix $\mathbf n=\mathbf n^\star$. The transmit power allocation problem
\begin{equation}
\max_{\mathbf p\ge0,\ \mathbf{1}^\mathsf{T}\mathbf p=P,\ \mathbf A\mathbf p\le\hat{\mathbf p}}
\sum_{i=1}^m \log\!\left(1+\frac{\beta_i p_i}{\sigma_i+n_i^\star}\right)
\end{equation}
is a strictly concave maximization over a polytope and therefore admits a unique optimizer. The KKT conditions are necessary and sufficient.

From the Lagrangian \eqref{eq:L-spatial}, the stationarity condition with respect to $p_i$ for each channel $i$ gives
\begin{equation}
\frac{\beta_i}{\sigma_i+n_i^\star+\beta_i p_i^\star}
=\nu^\star+(\mathbf A^\mathsf{T}\boldsymbol{\theta}^\star)_i-b_i^\star .
\end{equation}

If $p_i^\star>0$, then $b_i^\star=0$, yielding
\begin{equation}
p_i^\star
=
\frac{1}{\nu^\star+(\mathbf A^\mathsf{T}\boldsymbol{\theta}^\star)_i}
-\frac{\sigma_i+n_i^\star}{\beta_i}.
\end{equation}

If $p_i^\star=0$, the right-hand side is nonpositive. Combining both cases yields the water-filling expression.
\end{proof}

The linear constraints induce channel-dependent effective transmit power water levels for each active channel $i$
\begin{equation}
{\nu_A}_i^\star
\triangleq
\frac{1}{\nu^\star+(\mathbf A^\mathsf{T}\boldsymbol{\theta}^\star)_i}.
\end{equation}
When $\boldsymbol{\theta}^\star=0$, the global water level $\nu^\star$ is recovered, as shown in Fig.~\ref{fig:spatial_p}. The interference power subproblem
\begin{equation}
\min_{\mathbf n\ge 0,\ \mathbf{1}^\mathsf{T}\mathbf n=N}
\sum_{i=1}^m
\log\!\left(1+\frac{\beta_i p_i^\star}{\sigma_i+n_i}\right)
\end{equation}
is subject only to the simplex constraint. Consequently, the optimal interference
allocation admits the inverse water-filling solution given in (\ref{eq:n_op}). 

\begin{figure}
    \centering
    \includegraphics[width=0.95\linewidth]{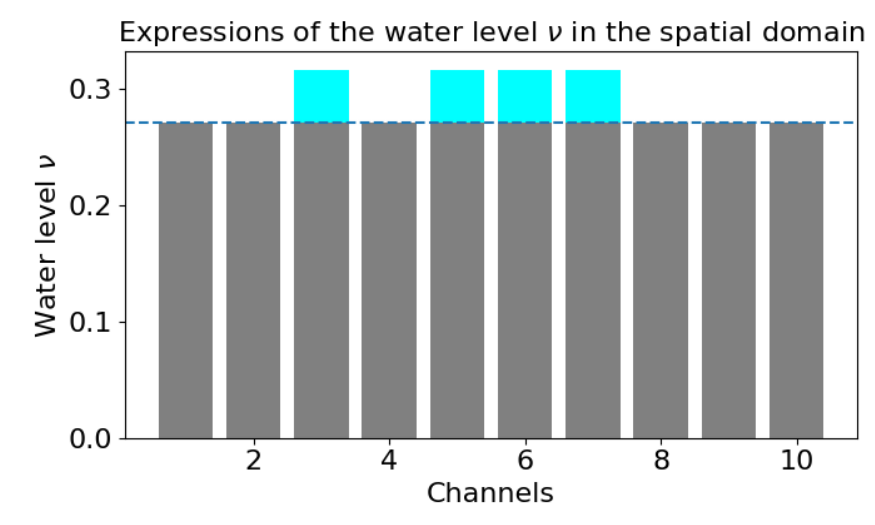}
    \color{black}
\caption{
Channel-wise effective water levels induced by the constraints. 
The total bar heights represent the inverse effective levels $1/{\nu_A}_i^\star$, while cyan segments correspond to the constraint-induced shift $(\mathbf A^\top\boldsymbol{\theta})_i$. 
The dashed line denotes the global water level $\nu^\star$.
}
\vspace{-3mm}
    \label{fig:spatial_p}
\end{figure}

\subsubsection*{PDHG on Gaussian Water-filling}
\label{subsec:primal_dual}
\color{black}
The adversarial water-filling problem over space in
\eqref{eq:spatial-awf} forms a saddle point problem that is
concave in the maximizing variable and convex in the minimizing
variable. Standard primal--dual methods such as PDHG provide a
useful interpretation of this structure. In particular, the dual
variables act as water levels enforcing the coupled resource
constraints, which motivates the learned primal--dual dynamics
developed later.
\color{black}

Consider the relaxation of a resource allocation problem
\begin{equation}
\min_{\mathbf{x}}\;
\sum_{i=1}^m f(x_i) + \mathbb{I}_C(\mathbf{x}),
\end{equation}
where $C = \left\{\mathbf{x} \in \mathbb{R}^m \mid \mathbf{1}^\mathsf{T}\mathbf{x} = b,\, \mathbf{x} \ge 0\right\}$.
Here $\mathbb{I}_C(\mathbf{x})$ is an indicator function of $C$, taking value $0$ if $\mathbf{x}\in C$ and $+\infty$ otherwise.
For the separable part $f(\mathbf{x})=\sum_{i=1}^m f(x_i)$, the proximal operator decomposes coordinate-wise for each channel $i$ with proximal stepsize $\tau>0$ \cite{parikh2014proximal}:
\begin{equation}
[\mathrm{prox}_{\tau f}(\mathbf{y})]_i
=
\arg\min_{x_i}
\left\{
f(x_i)
+
\frac{1}{2\tau}(x_i-y_i)^2
\right\}.
\end{equation}

Applying this framework to Gaussian water-filling, the power allocation subproblem for fixed interference $\mathbf{n}$ becomes
\begin{equation}
\min_{\mathbf{p}}
-
\sum_{i=1}^m
\log\!\left(1+\frac{\beta_i p_i}{\sigma_i+n_i}\right)
+
\mathbb{I}_{\tilde C}(\mathbf{p}),
\end{equation}
with $\tilde C=\{\mathbf{p}\mid \mathbf{p}\ge0,\ \mathbf{1}^\mathsf{T}\mathbf{p}=P\}$.
\color{black}
The total-power constraint couples the channel variables, and this
coupling is handled through its scalar dual variable, yielding
coordinate-wise primal proximal updates coupled through a common
dual variable that acts as the global water level.
\color{black}

Under the PDHG framework \cite{chambolle2011first,goldstein2013adaptive}, the spectrum-domain formulation can be written as 
\begin{equation}
\min_{\mathbf{p}\ge0}\max_{\nu}
-
\sum_{i=1}^m
\log\!\left(1+\frac{\beta_i p_i}{\sigma_i+n_i}\right)
+
\nu\!\left(\sum_{i=1}^m p_i-P\right).
\end{equation}

The corresponding PDHG updates take the form
\begin{align}
\nu^{(k+1)}
&=
\nu^{(k)}+\alpha\!\left(\sum_{i=1}^m p_i^{(k)}-P\right),\\
p_i^{(k+1)}
&=
\mathrm{prox}_{\tau f_i}\!
\left(
p_i^{(k)}-\tau\nu^{(k+1)}
\right),
\end{align}
where
$f(p_i)=-\log\!\left(1+\frac{\beta_i p_i}{\sigma_i+n_i}\right), \forall i$, and $\alpha,\tau>0$ denote the dual and primal stepsizes respectively.
The dual update corresponds to the proximal step associated with the conjugate of the constraint indicator and reduces to an affine update in this setting. The dual variable $\nu$ therefore parameterizes the global water level.

When additional linear constraints $\mathbf{A}\mathbf{p}\le\hat{\mathbf{p}}$ are present, the saddle problem introduces an additional dual variable $\boldsymbol{\theta}$. 
The resulting primal update becomes
\begin{equation}
p_i^{(k+1)}
=
\mathrm{prox}_{\tau f}
\!\left(
p_i^{(k)}
-
\tau
\big(
\nu^{(k+1)}+(\mathbf{A}^\mathsf{T}\boldsymbol{\theta}^{(k+1)})_i
\big)
\right), \quad \forall i
\end{equation}
showing that the effective water level is determined jointly by the global dual variable $\nu$ and the dual contribution $\mathbf{A}^\mathsf{T}\boldsymbol{\theta}$ associated with the additional linear constraints.

\subsection{Mercury/water-filling}

We now replace the Gaussian inputs with a general mutual-information objective.
Let
\begin{equation}
f(\mathbf p,\mathbf n)=
\sum_{i=1}^{m}
I_i\!\left(\frac{\beta_i p_i}{\sigma_i+n_i}\right),
\end{equation}
where $I_i(\gamma)$ denotes the mutual information (in nats/use) of channel $i$
at effective SINR $\gamma$. The AWF over space for channels under mercury/water-filling is
\begin{equation}
\max_{\mathbf p\ge 0,\ \mathbf{1}^\mathsf{T}\mathbf p=P,\ \mathbf A\mathbf p\le \hat{\mathbf p}}
\ \min_{\mathbf n\ge 0,\ \mathbf{1}^\mathsf{T}\mathbf n=N}
\sum_{i=1}^{m}
I_i\!\left(\frac{\beta_i p_i}{\sigma_i+n_i}\right).
\label{eq:spatial-awf-mercury}
\end{equation}

Under standard regularity conditions on $I(\cdot)$, we characterize stationary points of the minimax problem through the KKT conditions.
Introducing dual variables $\nu$ for $\mathbf{1}^\mathsf{T}\mathbf p=P$
and $\boldsymbol{\theta}\ge 0$ for $\mathbf A\mathbf p\le \hat{\mathbf p}$, Theorem~\ref{thm:spatial-pwf} extends directly to imply that, for any $i\in\{1,\dots,m\}$ with $p_i^\star>0$,
\begin{equation}
\frac{\beta_i}{\sigma_i+n_i^\star}
I_i'\!\left(\frac{\beta_i p_i^\star}{\sigma_i+n_i^\star}\right)
=
\nu^\star+(\mathbf A^\mathsf{T}\boldsymbol{\theta}^\star)_i.
\label{eq:p_mercury_spatial}
\end{equation}
Since the feasible set for interference power is unchanged,
for any $i\in\{1,\dots,m\}$ with $n_i^\star>0$, we obtain
\begin{equation}
\frac{\beta_i p_i^\star}{(\sigma_i+n_i^\star)^2}
I_i'\!\left(\frac{\beta_i p_i^\star}{\sigma_i+n_i^\star}\right)
=
-\mu^\star.
\label{eq:n_mercury_spatial}
\end{equation}
where $\mu^\star$ enforces $\mathbf{1}^\mathsf{T}\mathbf n=N$.
Using the I--MMSE identity as in \eqref{eq:mmse}, the AWF over space under mercury/water-filling admits the same water level interpretation as in the Gaussian case, with channel-dependent dual shifts induced by $(\mathbf A^\mathsf{T}\boldsymbol{\theta}^\star)_i$.

\subsubsection*{Mercury/water-filling and Proximal}
We now extend the proximal interpretation from Gaussian water-filling to mercury/water-filling. When the logarithmic objective is replaced by a general mutual-information function,
\begin{equation}
f(p_i)
=
-
I_i\!\left(
\frac{\beta_i p_i}{\sigma_i+n_i}
\right), \forall i
\end{equation}
the problem remains separable across channels and continues to admit a primal--dual formulation.

However, the key difference lies in the form of the proximal operator.
In the Gaussian case, the logarithmic objective leads to a rational derivative, yielding a quadratic optimality condition and hence a closed-form proximal update.

In contrast, under mercury/water-filling, the mutual information $I(\gamma)$ induces a nonlinear dependence through the I--MMSE relationship. 
As a result, the proximal operator
$\mathrm{prox}_{\tau f}(v)$ no longer admits a  closed-form solution. 
Instead, each update requires solving a one-dimensional nonlinear equation of the form
\begin{equation}
p_i + \tau f'(p_i) = v, \quad \forall i,
\end{equation}
where
\begin{equation}
f'(p_i)
=
-
\frac{\beta_i}{\sigma_i+n_i}
I_i'\!\left(
\frac{\beta_i p_i}{\sigma_i+n_i}
\right), \quad \forall i.
\end{equation}

This leads to an implicit update that must be computed numerically for each channel. 
While each subproblem is one-dimensional, each primal--dual iteration still requires the solution of multiple per-channel nonlinear equations. 
This can become computationally burdensome in large-scale systems or in scenarios with heterogeneous input constellations, where the MMSE functions differ across channels.
This loss of analytical tractability distinguishes mercury/water-filling from Gaussian water-filling and motivates incorporating an extragradient-style update within the foundation model in the next section.

\section{Foundation Model Architecture}
\label{sec:method}

The computational complexity of mercury/water-filling motivates replacing explicit proximal updates with learned primal--dual iterative dynamics. Unlike Gaussian water-filling, the associated updates generally do not admit closed-form solutions and instead require solving nonlinear implicit equations. This motivates a \textcolor{black}{domain-specific wireless} foundation model that learns the underlying AWF  dynamics in a data-driven manner and can be interpreted as a learned generalization of primal--dual first-order methods.
% \color{black}
% The intended generalization is across heterogeneous
% instances within the AWF family, rather than across arbitrary
% wireless optimization tasks.
% \color{black}

% [Modified 1] Add Gaussian special-case sentence in the opening.
The \textcolor{black}{domain-specific wireless} foundation model for AWF mirrors the KKT of adversarial water-filling: permutation-invariant encoding captures channel symmetry, graph neural message passing models sparse linear interactions induced by constraints, and learned primal--dual iterations approximate the coupled dynamics underlying water level coordination (cf. Fig.~\ref{fig:model}). Although primarily motivated by mercury/water-filling with discrete constellations, the framework ultimately applies to both Gaussian and mercury/water-filling settings, providing a unified and scalable approach. The Gaussian case is recovered as a special instance by choosing $I_i(\gamma)=\log(1+\gamma), \quad \forall i$.

\subsection{Architectural Motivation: Sets and Graphs}

The water-filling problem exhibits several properties that guide the foundation model design, including permutation invariance, variable problem dimension, constraint-induced coupling, and low-dimensional global coordination through the KKT system. In particular, the channel parameters $\{(\beta_i,\sigma_i)\}_{i=1}^m$ form an unordered set, motivating permutation-invariant representations.

To capture global channel interactions, we employ a Perceiver-style encoder \footnote{The Perceiver, introduced in \cite{jaegle2021perceivergeneralperceptioniterative}, maps variable-size inputs into a small latent array through cross-attention, followed by latent self-attention. It is adopted here to efficiently capture global channel dependencies.}. A small set of learnable latent tokens aggregates channel information through cross-attention and exchanges information internally \cite{jaegle2021perceivergeneralperceptioniterative}. The resulting fixed-size representation summarizes cross-channel competition while accommodating varying channel dimensions.
The linear constraint $\mathbf{Ap}\le\mathbf{\hat p}$ induces a natural bipartite graph between channel and constraint nodes, motivating a graph neural network (GNN)-based module. Through message passing, constraint nodes aggregate and propagate information among connected channels, approximating the coupling term $\mathbf A^\mathsf{T}\boldsymbol{\theta}$ in the KKT conditions. Since each constraint involves only channels corresponding to its nonzero entries, computation scales with graph connectivity. The KKT system further indicates that global coordination is governed by a small set of optimal dual variables acting as water levels, whose dynamics are parameterized through global latent representations.
Together, the Perceiver encoder and graph-based constraint propagation capture global interactions and sparse constraint couplings, supporting generalization across channel dimensions and constraint patterns. Alternative attention-based graph representations, including Graph Transformers, are evaluated in Section~VII-G.
\color{black}

\subsection{Problem Formulation}

We consider the AWF problem over space under mercury/water-filling introduced in Section~\ref{sec:spatial-awf}. 
\textcolor{black}{Given channel parameters $\{(\beta_i,\sigma_i)\}_{i=1}^m$, total power budgets $P,N$, 
and linear constraints $\mathbf A\mathbf p\le \hat{\mathbf p}$, 
the objective is the same as (\ref{eq:spatial-awf-mercury}). 
The Gaussian case is recovered as a special instance by choosing $I_i(\gamma)=\log(1+\gamma),\quad \forall i$.}

% [Modified 3] Clarify "active channels" in place.
For discrete constellations, the derivative $I_i'(\gamma)$ can be characterized via the I--MMSE identity~\cite{guo2005mutual}, which provides a convenient basis for gradient-based optimization. The KKT conditions imply that for all active channels $i$ with $p_i>0$,
\begin{equation}
I_i'(\gamma_i)\frac{\beta_i}{\sigma_i+n_i}
=
\nu + (\mathbf A^\mathsf{T}\boldsymbol{\theta})_i,
\label{eq:model}
\end{equation}
where $\nu$ and $\boldsymbol{\theta}$ are dual variables associated with the total budget and linear constraints and $\gamma_i = \frac{\beta_ip_i}{\sigma_i+n_i}, \quad \forall i$. Thus, coordination across channels is governed by a low-dimensional set of dual variables, with $\nu$ acting as a global water level and $\mathbf A^\top\boldsymbol{\theta}$ inducing channel-wise level shifts.

\subsection{Foundation Model for AWF}
We propose a \color{black}
domain-specific wireless \color{black}
foundation model for AWF that learns to emulate the underlying
solution dynamics, including the associated water-level search,
across varying channel dimensions, constraint patterns, and
modulation formats. Rather than learning a direct channel-to-power mapping for a fixed problem size, the model
\color{black}
learns reusable optimization dynamics shared across related AWF configurations.
This structured pretraining-and-reuse setting enables zero-shot reuse under varying problem configurations, distinguishing the proposed model from a fixed-size task-specific predictor.

% We propose a foundation model that learns to approximate the dynamics and, implicitly, the associated water level search across varying channel dimensions, constraint patterns, and modulation formats.
% Rather than learning a direct mapping from channels to powers, the model learns the algorithmic prior that generalizes across problem scales and distributions, forming a domain-specific foundation model for AWF problems.

\begin{figure*}
    \centering
    \includegraphics[width=0.92\linewidth]{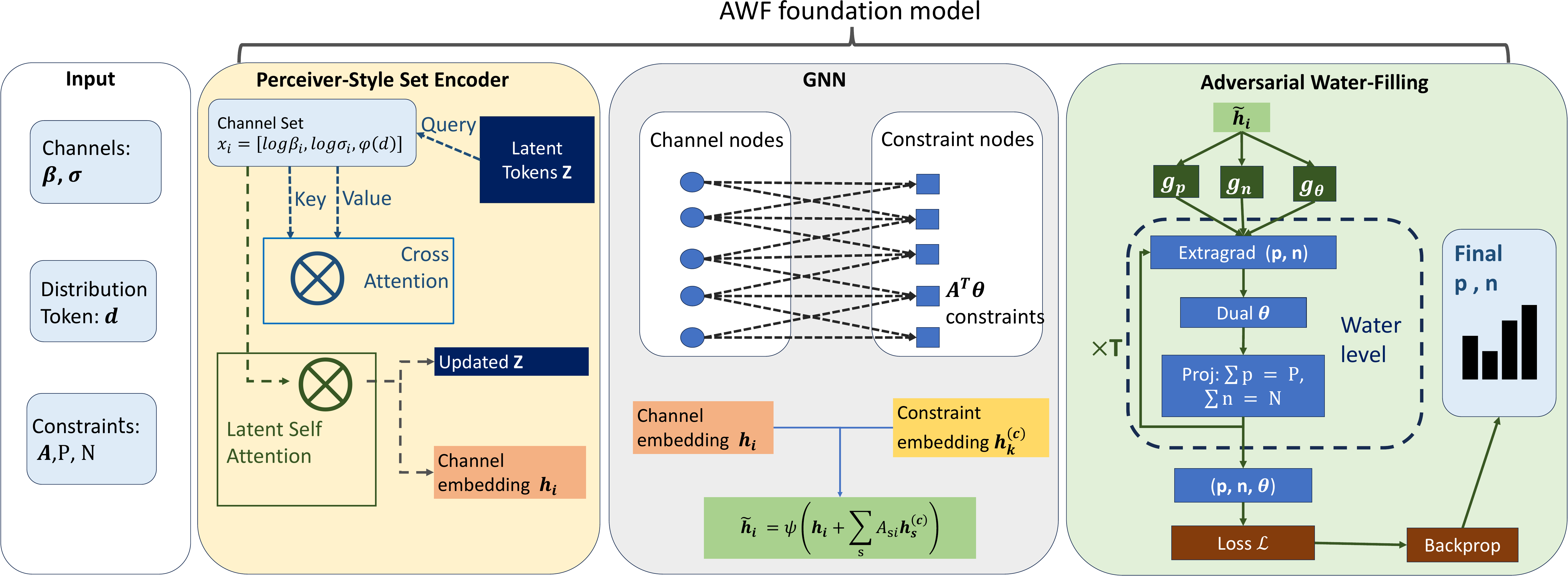}
\caption{
\textcolor{black}{Architecture of the AWF foundation model. A permutation-invariant encoder represents variable-size channel sets, a bipartite GNN captures constraint coupling, and learned projected primal--dual iterations produce the final allocations.}}
\vspace{-3mm}
    \label{fig:model}
\end{figure*}

\subsubsection{Distribution Token}

We sample $L$ logarithmically spaced SINR points $\{\gamma_j\}_{j=1}^{L}$ and construct a modulation-dependent feature vector
\begin{equation}
\mathbf d_i =
\left[
I_i'(\gamma_1), \dots, I_i'(\gamma_L),
I_i(\gamma_1), \dots, I_i(\gamma_L)
\right],
\end{equation}
which provides a compact numerical representation of the mutual-information curve associated with the modulation format assigned to channel $i$.
\color{black}
We refer to this feature vector as a distribution token, since it is treated as one input token by the learning model. 
Specifically, the entries $I(\gamma_j)$ describe the value of the mutual information at representative SINR points, while $I'(\gamma_j)$ captures the local variation of the curve. 
Thus, $\mathbf d$ encodes the modulation-dependent shape of the mutual-information function and is used as an input feature for the learning model.
The values of $I_i'(\gamma)$ are precomputed on an SINR grid, and values at arbitrary SINR arguments are obtained by interpolation.

\subsubsection{Set Encoding and Global Water Level Representation}

Each channel is represented by
\begin{equation}
\mathbf x_i =
\left[
\log \beta_i,\;
\log \sigma_i,\;
\phi(\mathbf d_i)
\right],
\end{equation}
where $\phi(\cdot)$ denotes a learned embedding of the distribution token. 
The unordered channel set $\{\mathbf x_i\}_{i=1}^m$ is processed by a permutation-invariant attention-based encoder
\begin{equation}
(\mathbf h, \mathbf Z) = \mathrm{SetEncoder}(\mathbf x),
\end{equation}
which produces channel embeddings $\mathbf h_i$ together with a small set of learnable latent tokens $\mathbf Z$.

The encoder uses a Perceiver-style architecture in which a small number of latent tokens $\mathbf Z$ first gather information from the channel set through cross-attention and then exchange information through latent self-attention. 
This latent bottleneck compresses a variable-size channel set into a fixed-size latent set $\mathbf Z$, thereby capturing cross-channel competition while remaining scalable as the number of channels varies. 
The latent set $\mathbf Z$ summarizes set-level statistics and coordinates channel interactions. 
From an optimization perspective, $\mathbf Z$ plays a role analogous to the low-dimensional dual state in the KKT system and helps capture the effective water-level adjustment that equalizes marginal utility across channels.

\subsubsection{Constraint-Aware Message Passing via GNN}
Linear constraints introduce interactions between subsets of channels. 
We model these interactions using a GNN constructed from the constraint matrix $\mathbf A \in \mathbb{R}^{S \times m}$. 
Specifically, we construct a bipartite graph with two types of nodes: channel nodes and constraint nodes. 
Each channel node represents one channel, or equivalently one optimization variable associated with that channel. 
Each constraint node represents one linear constraint, corresponding to one row of $\mathbf A$. 
An edge is added between channel node $i$ and constraint node $s$ whenever $A_{si} \neq 0$, indicating that channel $i$ appears in constraint $s$ with a nonzero coefficient. 
In this way, channels involved in the same linear constraint are connected through the corresponding constraint node, allowing the GNN to capture constraint-induced interactions among channels.

The GNN propagates information between channel nodes and constraint nodes through message passing. 
Message passing is performed in two stages.
First, each constraint node aggregates information from its incident channels to form a constraint embedding
\begin{equation}
\mathbf h^{(c)}_s =
\psi\!\left(
\sum_{i=1}^{m} A_{s i} \mathbf h_i
\right),
\end{equation}
where $\mathbf h_i$ denotes the channel embedding produced by the set encoder and $\psi(\cdot)$ is a learnable transformation. 
This aggregation summarizes how the channel representations contribute to each constraint.

Second, the constraint embeddings are propagated back to the channel nodes,
\begin{equation}
\tilde{\mathbf h}_i =
\psi\!\left(
\mathbf h_i + \sum_{s=1}^{S} A_{s i} \mathbf h^{(c)}_s
\right),
\end{equation}
so that each channel receives feedback from the constraints in which it participates.

This GNN-based message passing mechanism enables constraint-induced interactions to propagate across channels while preserving the structure of $\mathbf A$. 
From an optimization perspective, the feedback term approximates the coupling $\mathbf A^\top \boldsymbol{\theta}$ in the KKT conditions, where $\boldsymbol{\theta}$ denotes the dual variables associated with the linear constraints. 
Consequently, the graph module provides a scalable learnable approximation to the constraint-induced channel-wise coupling and the resulting level shifts.
% [Modified 9] Keep the update description consistent with the caption and wording.
\subsubsection{Primal--Dual Iterations and Water-Level Search}

Initial allocations are predicted from the channel embeddings
\begin{equation}
\mathbf p^{(0)} = P \cdot \mathrm{softmax}\!\big(g_p(\tilde{\mathbf h})\big),
\qquad
\mathbf n^{(0)} = N \cdot \mathrm{softmax}\!\big(g_n(\tilde{\mathbf h})\big).
\end{equation}

Starting from these initial allocations, the model performs $T$ learned primal--dual iterations that update the primal variables $(\mathbf p,\mathbf n)$ together with the dual variable $\boldsymbol{\theta}$ associated with the linear constraints. 
Let $\gamma_i = \frac{\beta_i p_i}{\sigma_i+n_i}$ denote the SINR of channel $i$. \color{black}
During iteration $t$, we use
$\gamma_i^{(t)}
=\beta_i p_i^{(t)}/(\sigma_i+n_i^{(t)})$
to denote the current SINR.
\color{black}
Using the I--MMSE identity, the gradient components of the saddle objective are given by
\begin{equation}
g_{p_i} =
I_i'(\gamma_i)\frac{\beta_i}{\sigma_i+n_i}
-
(\mathbf A^\top \boldsymbol{\theta})_i,
\end{equation}

\begin{equation}
g_{n_i} =
-
I_i'(\gamma_i)
\frac{\beta_i p_i}{(\sigma_i+n_i)^2},
\end{equation}
where $g_{p_i}$ and $g_{n_i}$ denote the gradient components with respect to $p_i$ and $n_i$, respectively. 
The dual gradient is given by
\begin{equation}
\mathbf g_{\theta} = \mathbf A\mathbf p - \hat{\mathbf p}.
\end{equation}

Let $\mathbf z=(\mathbf p,\mathbf n,\boldsymbol{\theta})$ denote the stacked primal--dual variables and let 
$\mathbf G(\mathbf z)=(\mathbf g_p,-\mathbf g_n,\mathbf g_\theta)$ denote the saddle-gradient field. 
The model performs a projected extragradient update
\begin{equation}
\mathbf z^{t+\frac12}
=
\Pi_{\mathcal X}\!\left(
\mathbf z^t + \mathbf D^t \mathbf G(\mathbf z^t)
\right),
\end{equation}
followed by the corrected step
\begin{equation}
\mathbf z^{t+1}
=
\Pi_{\mathcal X}\!\left(
\mathbf z^t + \mathbf D^t \mathbf G(\mathbf z^{t+\frac12})
\right),
\end{equation}
\color{black}
where $\mathbf D^t=\operatorname{blkdiag}
(\alpha_p^t\mathbf I_m,\alpha_n^t\mathbf I_m,
\alpha_\theta^t\mathbf I_S)$ is a learned positive block-scalar
stepsize matrix, so that all coordinates within each variable block
share the same learned stepsize.
Here, $\Pi_{\mathcal X}$ denotes the fixed deterministic blockwise feasibility projection onto the feasible product set
$\mathcal X=\Delta_P\times\Delta_N\times\mathbb R_{\ge0}^{S}$.
The associated operations contain no learnable parameters. For the primal variables, they preserve
nonnegativity and the prescribed budgets
$\mathbf 1^\mathsf{T}\mathbf p=P$ and
$\mathbf 1^\mathsf{T}\mathbf n=N$, while
$\boldsymbol{\theta}$ is restricted to the nonnegative orthant.
These projections implicitly adjust the effective water level across channels, analogous to the water level search in water-filling problems. The dual variable
$\boldsymbol{\theta}$ incorporates the linear constraint
$\mathbf A\mathbf p\le \hat{\mathbf p}$ through its constraint-aware
update, yielding an iteration
\textcolor{black}{that approximates the primal--dual dynamics of the AWF problem.}
\color{black}

Through training across heterogeneous instances, the foundation model learns an update rule that drives the iterates toward near-stationary water levels under varying problem sizes, constraint types, and modulation distributions.

\begin{algorithm}[t]
\caption{Foundation Model for Adversarial Water-Filling}
\label{alg:awf_foundation}
\begin{algorithmic}[1]

\REQUIRE Channel parameters $\{(\beta_i,\sigma_i)\}_{i=1}^m$, power budgets $(P,N)$, constraint matrix $(\mathbf A,\hat{\mathbf p})$, mutual-information tables $(I,I')$, iteration number $T$

\ENSURE Power allocations $(\mathbf p,\mathbf n)$

\STATE Construct $\mathbf d_i = [I_i'(\gamma_j), I_i(\gamma_j)]_{j=1}^{L}$ for each channel

\STATE Form channel features 
$\mathbf x_i = [\log\beta_i,\ \log\sigma_i,\ \phi(\mathbf d_i)]$

\STATE Compute channel embeddings and latent global state
$(\mathbf h,\mathbf Z) = \mathrm{SetEncoder}(\mathbf x)$

\STATE Compute constraint embeddings
$\mathbf h_s^{(c)} = \psi\!\left(\sum_{i=1}^{m} A_{si} \mathbf h_i\right)$

\STATE Update channel embeddings using the constraint GNN
$\tilde{\mathbf h}_i = \psi\!\left(\mathbf h_i + \sum_{s=1}^{S} A_{si} \mathbf h_s^{(c)}\right)$

\STATE Initialize allocations
$\mathbf p^{(0)} = P \cdot \mathrm{softmax}(g_p(\tilde{\mathbf h}))$,
$\mathbf n^{(0)} = N \cdot \mathrm{softmax}(g_n(\tilde{\mathbf h}))$

\STATE Initialize dual variable $\boldsymbol{\theta}^{(0)} = \mathbf 0$

\FOR{$t = 0,1,\dots,T-1$}

\STATE Compute SINR
$\gamma_i^{(t)} = \frac{\beta_i p_i^{(t)}}{\sigma_i+n_i^{(t)}}$

\STATE Evaluate $I'(\gamma_i^{(t)})$ 

\STATE Compute gradient components
$g_{p_i}^{(t)} =
I_i'(\gamma_i^{(t)})\frac{\beta_i}{\sigma_i+n_i^{(t)}}-(\mathbf A^\top\boldsymbol{\theta}^{(t)})_i$,
$g_{n_i}^{(t)} = -I_i'(\gamma_i^{(t)})
\frac{\beta_i p_i^{(t)}}{(\sigma_i+n_i^{(t)})^2}$,
$\mathbf g_{\theta}^{(t)}
=\mathbf A\mathbf p^{(t)}-\hat{\mathbf p}$

\STATE Compute $(\bar{\mathbf p},\bar{\mathbf n},\bar{\boldsymbol{\theta}})=
\Pi_{\mathcal X}\!\left(
(\mathbf p^{(t)},\mathbf n^{(t)},\boldsymbol{\theta}^{(t)})
+
\mathbf D^{(t)}
(\mathbf g_p^{(t)},-\mathbf g_n^{(t)},\mathbf g_\theta^{(t)})
\right)$

\STATE Compute corrected SINR
$\bar\gamma_i = \frac{\beta_i \bar p_i}{\sigma_i+\bar n_i}$

\STATE Evaluate corrected gradient components using
$(\bar{\mathbf p},\bar{\mathbf n},\bar{\boldsymbol{\theta}})$

$\bar g_{p_i} =
I_i'(\bar\gamma_i)\frac{\beta_i}{\sigma_i+\bar n_i}
-(\mathbf A^\top\bar{\boldsymbol{\theta}})_i$

$\bar g_{n_i} =
-I_i'(\bar\gamma_i)
\frac{\beta_i \bar p_i}{(\sigma_i+\bar n_i)^2}$,
$\bar{\mathbf g}_\theta
=\mathbf A\bar{\mathbf p}-\hat{\mathbf p}$

\STATE Update $(p^{(t+1)},n^{(t+1)},\theta^{(t+1)}) =
\Pi_{\mathcal X}\!\left(
(p^{(t)},n^{(t)},\theta^{(t)})
+
D^{(t)}(\bar g_p,-\bar g_n,\bar g_\theta)
\right)$

% \STATE Enforce constraints
% $p^{(t+1)} \leftarrow \Pi_{\Delta_P}(p^{(t+1)})$,
% $n^{(t+1)} \leftarrow \Pi_{\Delta_N}(n^{(t+1)})$

% % \STATE Update dual variable
% % $\theta^{(t+1)} \leftarrow [\theta^{(t+1)} + D_\theta^{(t)} (Ap^{(t+1)}-\hat p)]_+$
% \STATE Project the dual variable $\theta^{(t+1)} \leftarrow \max(\theta^{(t+1)},0)$.

\ENDFOR

\RETURN $(\mathbf p^{(T)},\mathbf n^{(T)})$

\end{algorithmic}
\end{algorithm}

\subsection{Training Objective}

We train on random AWF instances across varying channel dimensions, sparse constraint graphs, and modulation types. Define the normalized mutual-information objective
\begin{equation}
J(\mathbf p,\mathbf n)
\triangleq
\frac{1}{m}
\sum_{i=1}^{m}
I_i\!\left(\frac{\beta_i p_i}{\sigma_i+n_i}\right).
\end{equation}

The training objective is motivated by primal--dual optimality conditions for minimax problems and by PDHG-type algorithms~\cite{chambolle2011first,goldstein2013adaptive}. 
Rather than solving each instance to optimality, we measure how close the current iterate $(\mathbf p,\mathbf n)$ is to a stationary point of the minimax problem using a tractable surrogate of the saddle residual based on a few projected ascent and descent steps. Specifically,
\begin{equation}
\mathbf p^{\uparrow}
\approx
\mathrm{Ascent}(\mathbf p;\mathbf n),
\qquad
\mathbf n^{\downarrow}
\approx
\mathrm{Descent}(\mathbf n;\mathbf p),
\end{equation}
where $\mathrm{Ascent}(\cdot)$ denotes a few projected updates that increase $J(\mathbf p,\mathbf n)$ with respect to $\mathbf p$, while $\mathrm{Descent}(\cdot)$ denotes a few projected updates that decrease $J(\mathbf p,\mathbf n)$ with respect to $\mathbf n$.

The resulting relative residual is
\begin{equation}
\mathcal{L}_{\mathrm{gap}} =
\frac{J(\mathbf p^{\uparrow}, \mathbf n) - J(\mathbf p, \mathbf n^{\downarrow})}
{|J(\mathbf p,\mathbf n)| + \epsilon}.
\end{equation}

To promote feasibility with respect to the linear constraints, we introduce a smooth penalty
\begin{equation}
\mathcal{L}_{\mathrm{ineq}} =
\left\|\max(\mathbf A\mathbf p-\hat{\mathbf p},\mathbf 0)\right\|_2^2 .
\end{equation}

We further incorporate a stationarity regularizer derived from the KKT optimality conditions on active channels,
\begin{equation}
\mathcal{L}_{\mathrm{kkt}}
=
\frac{1}{|\mathcal{A}_{+}|}
\sum_{i\in\mathcal{A}_{+}}
\left|
I_i'(\gamma_i)\frac{\beta_i}{\sigma_i+n_i}
-(\mathbf A^\top\boldsymbol{\theta})_i
-\nu
\right|,
\end{equation}
where $\mathcal{A}_{+}=\{i:p_i>\tau\}$ denotes the active channel set and $\nu$ is the global water level associated with the simplex constraint, as characterized by \eqref{eq:model}.

\color{black}
The overall training objective combines the residual loss, feasibility and KKT-based stationarity regularization.
\color{black}
\begin{equation}
\mathcal{L}
=
\mathcal{L}_{\mathrm{gap}}
+
\lambda_{\mathrm{ineq}}\mathcal{L}_{\mathrm{ineq}}
+
\lambda_{\mathrm{kkt}}\mathcal{L}_{\mathrm{kkt}} .
\end{equation}

The model is trained over heterogeneous AWF instances with varying channel dimensions, constraint graphs, and modulation distributions. 
As a result, the learned model generalizes not only across problem size and constraint types but also across distributional families, effectively learning the water level rather than a fixed-size regression mapping.

\section{Theoretical Analysis}
\label{sec:model_theory}
\color{black}
This section analyzes the learned projected extragradient
dynamics associated with the proposed AWF formulation. We focus
on their projected-stationarity consistency and conditional
local stability under the stated regularity assumptions. Applying these standard tools to the learned AWF
dynamics requires accounting for active-set-dependent
projection differentiability and the generally nonconvex
mercury/water-filling interference block.
\color{black}

Define the mutual-information objective
\begin{equation}
F(\mathbf p,\mathbf n)
=
\sum_{i=1}^{m} I_i(\gamma_i),
\qquad s.t. \quad
\gamma_i
=
\frac{\beta_i p_i}{\sigma_i + n_i}.
\end{equation}
We assume that $\sigma_i>0$ and $\beta_i>0$ for all $i$.
The primal feasible sets are $\Delta_P
=
\{\mathbf p\in\mathbb{R}^m_{\ge0}:\mathbf{1}^\mathsf{T}\mathbf p=P\}$, $\Delta_N
=
\{\mathbf n\in\mathbb{R}^m_{\ge0}:\mathbf{1}^\mathsf{T}\mathbf n=N\}$, together with the constraint $\mathbf A\mathbf p\le \hat{\mathbf p}$, $\mathbf A\in\mathbb{R}_{\ge0}^{S\times m}$.
Introduce the partial Lagrangian
\begin{equation}
\mathcal{L}(\mathbf p,\mathbf n,\boldsymbol{\theta})
=
F(\mathbf p,\mathbf n)
-\langle \boldsymbol{\theta},
\mathbf A\mathbf p-\hat{\mathbf p}\rangle,
\qquad
\boldsymbol{\theta}\in\mathbb{R}_{\ge0}^S .
\end{equation}

% [Modified] Use I_i' as the main notation and mention I--MMSE as a relation.
For discrete constellations, $I_i'(\gamma)$ can be related to the MMSE function through the I--MMSE identity. 
The gradients of the Lagrangian are
\begin{equation}
\nabla_{p_i}\mathcal{L}(\mathbf p,\mathbf n,\boldsymbol{\theta})
=
I_i'(\gamma_i)\frac{\beta_i}{\sigma_i+n_i}
-
(\mathbf A^\mathsf{T}\boldsymbol{\theta})_i,
\label{eq:grad_p_L}
\end{equation}
\begin{equation}
\nabla_{n_i}\mathcal{L}(\mathbf p,\mathbf n,\boldsymbol{\theta})
=
-I_i'(\gamma_i)
\frac{\beta_i p_i}{(\sigma_i+n_i)^2},
\label{eq:grad_n_L}
\end{equation}
and
\begin{equation}
\nabla_{\boldsymbol{\theta}}
\mathcal{L}(\mathbf p,\mathbf n,\boldsymbol{\theta})
=
-(\mathbf A\mathbf p-\hat{\mathbf p}).
\label{eq:grad_theta_L}
\end{equation}

Let $\mathbf z=(\mathbf p,\mathbf n,\boldsymbol{\theta})$,
and define the product constraint set
$\mathcal{X}
=
\Delta_P\times \Delta_N\times \mathbb{R}_{\ge0}^S.$
The saddle-gradient field is
\begin{equation}
\mathbf G(\mathbf z)
=
\big(
\nabla_{\mathbf p}\mathcal L(\mathbf z),\;
-\nabla_{\mathbf n}\mathcal L(\mathbf z),\;
-\nabla_{\boldsymbol{\theta}}\mathcal L(\mathbf z)
\big).
\label{eq:saddle_field}
\end{equation}

The learned model performs projected extragradient iterations of the form
\begin{equation}
\mathbf z^{k+\frac{1}{2}}
=
\Pi_{\mathcal X}\!\left(
\mathbf z^k + \mathbf D^k \mathbf G(\mathbf z^k)
\right),
\label{eq:eg_half}
\end{equation}
\begin{equation}
\mathbf z^{k+1}
=
\Pi_{\mathcal X}\!\left(
\mathbf z^k + \mathbf D^k \mathbf G(\mathbf z^{k+\frac12})
\right),
\label{eq:model_extragrad_map}
\end{equation}
% [Modified] Replace K by S in the stepsize matrix.
\color{black}
where $\mathbf D^k
=
\mathrm{blkdiag}(\alpha_p^k \mathbf I_m,
\alpha_n^k \mathbf I_m,
\alpha_\theta^k \mathbf I_S)$
is the positive block-scalar stepsize matrix
defined above predicted by the network and clamped to a bounded interval. 
\color{black}
% where $D^k
% =
% \mathrm{diag}(\alpha_p^k I_m,\alpha_n^k I_m,\alpha_\theta^k I_S)$ is a positive diagonal stepsize matrix predicted by the network and clamped to a bounded interval.

\color{black}
\begin{proposition}[Projected-stationarity consistency]
\label{prop:kkt_consistency}

Assume that each $I_i(\gamma)$ is continuously differentiable
over the considered SINR interval. Let
$\mathbf z^\star
=
(\mathbf p^\star,\mathbf n^\star,
\boldsymbol{\theta}^\star)
\in\mathcal X$,
and define
\begin{equation}
\mathbf D^\star
=
\operatorname{blkdiag}
\left(
\alpha_p^\star\mathbf I_m,\,
\alpha_n^\star\mathbf I_m,\,
\alpha_\theta^\star\mathbf I_S
\right),
\end{equation}
where
$\alpha_p^\star,\alpha_n^\star,\alpha_\theta^\star>0$.

If $\mathbf z^\star$ satisfies the first-order stationarity condition
\begin{equation}
\mathbf z^\star
=
\Pi_{\mathcal X}
\left(
\mathbf z^\star+
\mathbf D^\star
G(\mathbf z^\star)
\right),
\label{eq:projected_stationarity}
\end{equation}
then $\mathbf z^\star$ satisfies the first-order KKT conditions
of the AWF problem.
\end{proposition}

\begin{proof}
By the optimality condition of Euclidean projection,
\eqref{eq:projected_stationarity} implies
\begin{equation}
\mathbf D^\star G(\mathbf z^\star)
\in
N_{\mathcal X}(\mathbf z^\star),
\end{equation}
where $N_{\mathcal X}$ denotes the normal cone.

Since $\mathbf D^\star$ acts as multiplication by a strictly positive scalar within each of the $\mathbf p$-, $\mathbf n$-, and $\boldsymbol{\theta}$-blocks, the corresponding blockwise normal-cone inclusions are invariant to
this positive scaling. Hence,
\begin{align}
\nabla_{\mathbf p}
L(\mathbf z^\star)
&\in
N_{\Delta_P}(\mathbf p^\star),\\
-\nabla_{\mathbf n}
L(\mathbf z^\star)
&\in
N_{\Delta_N}(\mathbf n^\star),\\
-\nabla_{\boldsymbol{\theta}}
L(\mathbf z^\star)
&\in
N_{\mathbb R_+^S}
(\boldsymbol{\theta}^\star).
\end{align}
The last inclusion is equivalent to primal feasibility, dual feasibility, and complementary slackness for $\mathbf A\mathbf p\leq\hat{\mathbf p}$. These relations are precisely the first-order KKT conditions.
\end{proof}

Proposition~\ref{prop:kkt_consistency} establishes that the projected first-order stationarity condition in \eqref{eq:projected_stationarity} implies the KKT conditions of the AWF problem, while the local stability of the learned extragradient dynamics is analyzed separately below.
\color{black}

\begin{theorem}[Foundation Model AWF]
\label{thm:water_level}
Consider the simplex projection step in Line~15 of Algorithm~\ref{alg:awf_foundation}. 
Let $(\mathbf p^\star,\mathbf n^\star,\boldsymbol{\theta^\star})$ be a KKT point of the AWF. 
Then for every active channel $i\in\{1,\dots,m\}$ with $p_i^\star>0$, there exists a scalar $\nu^\star$ associated with the simplex constraint $\mathbf{1}^\mathsf{T}\mathbf p=P$ such that
\begin{equation}
I_i'(\gamma_i^\star)
\frac{\beta_i}{\sigma_i+n_i^\star}
=
\nu^\star + (\mathbf A^\mathsf{T}\boldsymbol\theta^\star)_i.
\end{equation}
\end{theorem}

\begin{proof}
The KKT condition for the constrained maximization in $\mathbf p$ implies the existence of dual variable $\nu^\star$ for the simplex constraint $\mathbf{1}^\mathsf{T}\mathbf p=P$ and nonnegative $b^\star$ for $\mathbf p\ge0$ such that $\nabla_{p_i}\mathcal L(\mathbf p^\star,\mathbf n^\star,\boldsymbol{\theta^\star})-\nu^\star+b_i^\star=0$.
If $p_i^\star>0$, complementary slackness yields $b_i^\star=0$. Substituting \eqref{eq:grad_p_L} gives
\begin{equation}
I_i'(\gamma_i^\star)
\frac{\beta_i}{\sigma_i+n_i^\star}
=
\nu^\star + (\mathbf A^\mathsf{T} \boldsymbol\theta^\star)_i,
\end{equation}
which proves the claim for every active channel $i\in\{1,\dots,m\}$ with $p_i^\star>0$.
\end{proof}

Under mercury/water-filling, the mutual-information objective is not guaranteed to be globally convex in $\mathbf n$. 
Even when $I_i(\cdot)$ is smooth, the mapping
\begin{equation}
n_i
\mapsto
I_i\!\left(\frac{\beta_i p_i}{\sigma_i+n_i}\right), \quad \forall i
\end{equation}
can be nonconvex in $n_i$, depending on the constellation and the SNR regime. Accordingly, convex--concave guarantees for this minimax problem do not apply in general.
\textcolor{black}{
We next examine the local stability of the learned update map under explicit regularity conditions.
}

\color{black}
\begin{assumption}[Local regularity and stability]
\label{assump:local_regularity}
Let $\mathcal T$ denote the learned projected extragradient
update map induced by
\eqref{eq:eg_half}--\eqref{eq:model_extragrad_map},
with the learned parameters fixed after training. Let
$\mathbf z^\star$ be a fixed point of $\mathcal T$.
Assume that there exists a neighborhood $\mathcal U$ of
$\mathbf z^\star$ such that:

\begin{enumerate}
    \item[(i)] the saddle-gradient field $G$ is
    $L_G$-Lipschitz continuous on $\mathcal U$; and

    \item[(ii)] the active sets of the simplex and
    nonnegative-orthant projections remain unchanged in
    $\mathcal U$, and the projected update map $\mathcal T$
    is continuously differentiable on $\mathcal U$ relative
    to the corresponding feasible active manifold; and

    \item[(iii)] the Jacobian of the learned update map at
    $\mathbf z^\star$ satisfies
    \begin{equation}
    \left\|
    J_{\mathcal T}(\mathbf z^\star)
    \right\|_2 < 1 .
    \label{eq:local_contraction}
    \end{equation}
\end{enumerate}
\end{assumption}

Assumption~\ref{assump:local_regularity} provides sufficient
local conditions for the stability analysis below.
Condition~(ii) restricts the analysis to a locally fixed active
manifold and therefore excludes projection switching points,
where simplex and nonnegative-orthant projections are
generally nondifferentiable. Condition~(iii) imposes local
contractivity on the trained update map.

\begin{theorem}[Conditional local convergence and computational complexity]
\label{thm:local_convergence}
Consider the learned AWF update map $\mathbf z^{k+1}
    =
    \mathcal T(\mathbf z^k).$
Under Assumption~\ref{assump:local_regularity}, there exist
constants $c\in(0,1)$ and $\delta>0$ such that, whenever $\|\mathbf z^0-\mathbf z^\star\|_2\leq\delta$,
the iterates remain in $\mathcal U$ and satisfy
\begin{equation}
    \|\mathbf z^{k+1}-\mathbf z^\star\|_2
    \leq
    c\,
    \|\mathbf z^k-\mathbf z^\star\|_2.
    \label{eq:local_linear_convergence}
\end{equation}
Hence, the learned projected extragradient dynamics converge
locally linearly to $\mathbf z^\star$ under the stated
conditions.

In addition, one inference forward pass of the model has
computational complexity
\begin{equation}
\mathcal O\bigl(L_{\rm att}md^2+Ed+T(m\log m+E)\bigr)
\end{equation}
and memory complexity
\begin{equation}
\mathcal O(md+E).
\end{equation}
\end{theorem}

\begin{proof}
By Assumption~\ref{assump:local_regularity}(ii),
$\mathcal T$ is continuously differentiable in a neighborhood
of $\mathbf z^\star$ relative to the locally fixed active
manifold. Moreover, $\|J_{\mathcal T}(\mathbf z^\star)\|_2 < 1$
by Assumption~\ref{assump:local_regularity}(iii).
By continuity of the Jacobian, there exist
$c\in(0,1)$ and a sufficiently small convex neighborhood
$\mathcal U_0\subseteq\mathcal U$, relative to the locally
fixed active manifold, such that
\begin{equation}
\sup_{\mathbf z\in\mathcal U_0}
\|J_{\mathcal T}(\mathbf z)\|_2
\leq c.
\end{equation}
By the mean-value inequality on this local manifold,
$\mathcal T$ is contractive on $\mathcal U_0$, and hence
\begin{equation}
\|\mathcal T(\mathbf z)-\mathbf z^\star\|_2
\leq
c\|\mathbf z-\mathbf z^\star\|_2.
\end{equation}
The local linear convergence in
\eqref{eq:local_linear_convergence} follows.

For the computational complexity, the Perceiver-style encoder
processes $m$ channel tokens of dimension $d$ over
$L_{\rm att}$ layers in
$\mathcal O(L_{\rm att}md^2)$ operations. Sparse constraint
message passing requires $\mathcal O(Ed)$ operations, and each
projected extragradient step requires
$\mathcal  O(m\log m+E)$ operations. Summing over $T$ unrolled steps
gives the stated complexity. The corresponding memory
complexity is $\mathcal O(md+E)$.
\end{proof}

% [Modified] Make the concluding claim slightly more precise.
Proposition~\ref{prop:kkt_consistency} and
Theorem~\ref{thm:local_convergence} characterize complementary
properties of the learned AWF dynamics. The proposition links
projected first-order stationarity to the KKT system of the
underlying AWF problem, whereas the theorem establishes local
linear stability of the learned update map under
Assumption~\ref{assump:local_regularity}. For the generally nonconvex mercury/water-filling formulation,
these results provide first-order and local dynamical
characterizations under the stated assumptions.
\color{black}

\section{Experiments}

We evaluate the foundation model on the discrete-constellation AWF problem under varying channel dimensions and constraints.

\color{black}
\subsection{Problem Setting and Evaluation Protocol}
\color{black}
In the channels under mercury/water-filling, the mutual information depends on discrete constellations rather than the Gaussian logarithmic form. We use 16QAM and 64QAM during training, and obtain $I'(\gamma)=\mathrm{mmse}(\gamma)$ from precomputed interpolation tables on a logarithmically spaced SNR grid over approximately $[-10,30]$ dB. Channel gains and noise powers are sampled as
$\log\beta_i\sim\mathcal N(0,0.6^2)$ and
$\log\sigma_i\sim\mathcal N(0,0.25^2)$.
The per-channel transmit and interference budgets are sampled as
$\bar P,\bar N\sim\mathcal U(0.05,0.5)$, and the total budgets are set to
$P=m\bar P$ and $N=m\bar N$. Linear constraints are generated from sparse nonnegative matrices
$\mathbf A\in\mathbb R^{S\times m}$, where
$S=\max\{1,\lfloor \rho m\rceil\}$ and
$\rho\sim\mathcal U(0.05,0.30)$.
Each row of $\mathbf A$ is normalized to unit sum, and $\hat p$ is generated by adding positive random slack to $\mathbf{Ap}$ for a sampled feasible power allocation.
Training uses $m\in\{32,64,128,256,512\}$, while evaluation additionally includes $m=16$ and the unseen large-scale case $m=1024$.
For modulation generalization, 256QAM is used only as an unseen held-out evaluation format.

The model is trained with Adam using learning rate $10^{-4}$ and batch size $16$, with $T=64$ unrolled projected extragradient updates.
Training instances are generated online. Runtime is measured after warm-up over repeated runs on a single NVIDIA GeForce RTX 3050 GPU.

\color{black}
We evaluate the pretrained AWF foundation model using the
normalized objective
\begin{equation}
J(\mathbf p,\mathbf n)
=
\frac{1}{m}
\sum_{i=1}^{m}
I_i(\gamma_i),
\end{equation}
the average inequality violation
\begin{equation}
V(\mathbf p)
=
\frac{1}{S}
\sum_{j=1}^{S}
\max
((\mathbf A\mathbf p)_j-\hat p_j
,0),
\end{equation}
and a projected first-order stationarity residual
\begin{equation}
R_{\rm stat}(\mathbf z)
=
\frac{
\left\|
\mathbf z-
\Pi_{\mathcal X}
\!\left(
\mathbf z+\eta_{\rm eval}\mathbf G(\mathbf z)
\right)
\right\|_2
}{
\eta_{\rm eval}\sqrt{2m+S}
},
\qquad
\eta_{\rm eval}=0.01,
\label{eq:eval_rstat}
\end{equation}
where $\mathbf z=(\mathbf p,\mathbf n,\boldsymbol{\theta})$ and $\mathbf G$ is the saddle-gradient field defined in Section VI. 
The objective value is reported as a performance measure,
whereas first-order solution quality is characterized by
$R_{\rm stat}$ together with constraint feasibility.

\color{black}
For runtime evaluation, we further use a per-instance
matched-quality protocol based jointly on projected
stationarity and constraint feasibility. For each test instance,
let $R_{\rm stat}^{\rm AWF}$ and $V_{\rm AWF}$ denote the
quality attained by the fixed-depth pretrained AWF model.
The matching iteration of Mirror-Prox is defined as
\begin{equation}
k_{\rm match}
=
\min
\left\{
k:
\begin{array}{l}
R_{\rm stat}^{\rm MP}(k)
\le
R_{\rm stat}^{\rm AWF},\\[1mm]
V_{\rm MP}(k)
\le
\max\{V_{\rm AWF},10^{-3}\}
\end{array}
\right\}, 
\label{eq:matched_quality}
\end{equation}
where $10^{-3}$ serves as a numerical feasibility tolerance. Thus, the runtime comparison is performed at matched
first-order stationarity and constraint feasibility on each
instance. If no such iterate is reached within
$K_{\max}=3000$, the instance is reported as unmatched.
\color{black}

\color{black}
For each successfully matched instance $i$, the corresponding
runtime speedup is ${U}_i
=
\frac{
t_{{\rm MP},i}(k_{{\rm match},i})
}{
t_{{\rm AWF},i}
}$.
Speedup statistics are reported over the matched instances,
together with the corresponding matching rate.
\color{black}
\color{black}
The matching iteration is identified in a separate evaluation
pass. Mirror-Prox is then rerun for the fixed
$k_{\rm match}$ iterations for timing, without evaluating
stationarity metrics inside the timed region.
\color{black}

\subsection{Baseline Methods}

As a baseline, we implement the Mirror-Prox method \cite{nemirovski2004prox}, a standard extragradient algorithm for the minimax problems and variational inequalities. 
\color{black}
For reference, the standard projected Mirror-Prox iteration
can be written as

\begin{equation}
z^{k+\frac12} =
\Pi_{\mathcal{X}}\left(z^k + \eta\,G(z^k)\right),
\end{equation}

\begin{equation}
z^{k+1} =
\Pi_{\mathcal{X}}\left(z^k + \eta\,G(z^{k+\frac12})\right),
\end{equation}
where $\mathbf z=(\mathbf p,\mathbf n,\boldsymbol{\theta})$
and $\mathbf G$ is the saddle-gradient field defined in Section~VI. 
Our numerical implementation uses fixed positive scalar stepsizes for the three variable blocks,
$\alpha_p=0.01$, $\alpha_n=0.03$, and
$\alpha_\theta=0.05$, applied in both Mirror–Prox steps.
% The interference block is parameterized by an unconstrained
% auxiliary variable followed by softplus normalization to enforce
% $\mathbf1^\mathsf{T}\mathbf n=N$.
\color{black}

\color{black}
As a low-complexity learning-based baseline, we additionally
consider Direct-GNN. It uses the same channel and constraint
representations, Perceiver encoder, and constraint-aware
message-passing backbone as the proposed AWF foundation model,
but replaces the $T$-step learned primal--dual dynamics with
a one-shot prediction head that directly produces the resource
allocation variables. The resulting model has a parameter
count comparable to that of the full AWF model, allowing us to distinguish the effect
of the learned iterative structure from that of neural
representation capacity alone. Direct-GNN is trained from
scratch using the same instance generator, training
distribution, and training protocol as the full model.
Its detailed comparison is reported together with the
architectural ablations in Section~VII-G.
\color{black}

\color{black}
\subsection{Experiment I: Size Generalization}

\begin{table}[t]
\color{black}
\centering
\caption{
Size generalization across channel dimensions.
}
\label{tab:size_generalization}
\setlength{\tabcolsep}{3.4pt}
\renewcommand{\arraystretch}{0.88}
\begin{tabular}{c l c c c}
\toprule
$m$ & Method & $J$ & $V$ & $R_{\rm stat}$ \\
\midrule
16
& Direct-GNN  & 0.4810 & $4.41{\times}10^{-3}$ & 0.2972 \\
& AWF Model   & 0.4786 & $1.91{\times}10^{-5}$ & 0.0357 \\
& Mirror-Prox & 0.4797 & $4.61{\times}10^{-4}$ & 0.0290 \\
\midrule
32
& Direct-GNN  & 0.5423 & $5.48{\times}10^{-3}$ & 0.3220 \\
& AWF Model   & 0.5479 & $7.77{\times}10^{-5}$ & 0.0521 \\
& Mirror-Prox & 0.5498 & $5.37{\times}10^{-4}$ & 0.0356 \\
\midrule
64
& Direct-GNN  & 0.5075 & $9.71{\times}10^{-3}$ & 0.3122 \\
& AWF Model   & 0.5102 & $3.09{\times}10^{-4}$ & 0.0469 \\
& Mirror-Prox & 0.5113 & $9.33{\times}10^{-4}$ & 0.0352 \\
\midrule
128
& Direct-GNN  & 0.4927 & $7.86{\times}10^{-3}$ & 0.3409 \\
& AWF Model   & 0.5154 & $7.37{\times}10^{-4}$ & 0.0562 \\
& Mirror-Prox & 0.5172 & $1.00{\times}10^{-3}$ & 0.0513 \\
\midrule
256
& Direct-GNN  & 0.4772 & $7.72{\times}10^{-3}$ & 0.3131 \\
& AWF Model   & 0.4973 & $1.24{\times}10^{-3}$ & 0.0452 \\
& Mirror-Prox & 0.4991 & $1.35{\times}10^{-3}$ & 0.0433 \\
\midrule
512
& Direct-GNN  & 0.4926 & $6.71{\times}10^{-3}$ & 0.2837 \\
& AWF Model   & 0.5023 & $1.38{\times}10^{-3}$ & 0.0303 \\
& Mirror-Prox & 0.5032 & $1.50{\times}10^{-3}$ & 0.0242 \\
\midrule
1024
& Direct-GNN  & 0.5011 & $4.87{\times}10^{-3}$ & 0.3132 \\
& AWF Model   & 0.5310 & $2.09{\times}10^{-3}$ & 0.0489 \\
& Mirror-Prox & 0.5325 & $1.88{\times}10^{-3}$ & 0.0535 \\
\bottomrule
\end{tabular}
\end{table}
We first evaluate generalization across channel dimensions.
For each dimension, predefined test instances are generated
from the same channel, modulation, and constraint distributions
as in training. The pretrained AWF foundation model is kept
fixed across all dimensions without instance-specific
fine-tuning. Training uses
$m\in\{32,64,128,256,512\}$, while evaluation additionally
includes $m=16$ and the unseen large-scale case $m=1024$. In addition to Mirror-Prox, we include Direct-GNN as a
learning-based baseline. Direct-GNN uses a comparable neural
representation backbone but directly predicts the resource
allocations in one forward pass, without the learned
primal--dual rollout. All methods are evaluated using the
normalized objective $J$, average inequality violation $V$, and
projected first-order stationarity residual $R_{\rm stat}$
defined in Section~VII-A.

Table~\ref{tab:size_generalization} shows that AWF maintains substantially smaller stationarity residuals than Direct-GNN across all tested dimensions, with low constraint violation. Its average $R_{\rm stat}$ ranges from $0.0303$ to $0.0562$ over $m=16$--$1024$, versus $0.2837$--$0.3409$ for Direct-GNN, while objective values remain close to those of Mirror-Prox. Despite its lower raw forward latency, Direct-GNN exhibits markedly larger stationarity and feasibility errors, highlighting the importance of comparing solution quality alongside runtime. This supports the learned primal--dual rollout in preserving first-order quality across dimensions. Fig.~\ref{fig:matched_size_runtime} reports the matched-quality runtime comparison defined in Section~VII-A. Mirror-Prox reaches the AWF stationarity--feasibility target within $K_{\max}=3000$ iterations for 142 of 168 instances ($84.5\%$), over which AWF achieves a median speedup of $13.84\times$. At unseen $m=1024$, 16 of 24 cases are matched, with a median speedup of $23.44\times$. These results demonstrate reduced time to comparable first-order quality rather than uniformly superior final stationarity.
\color{black}

\begin{figure}
    \centering
    \color{black}
    \includegraphics[width=1\linewidth]{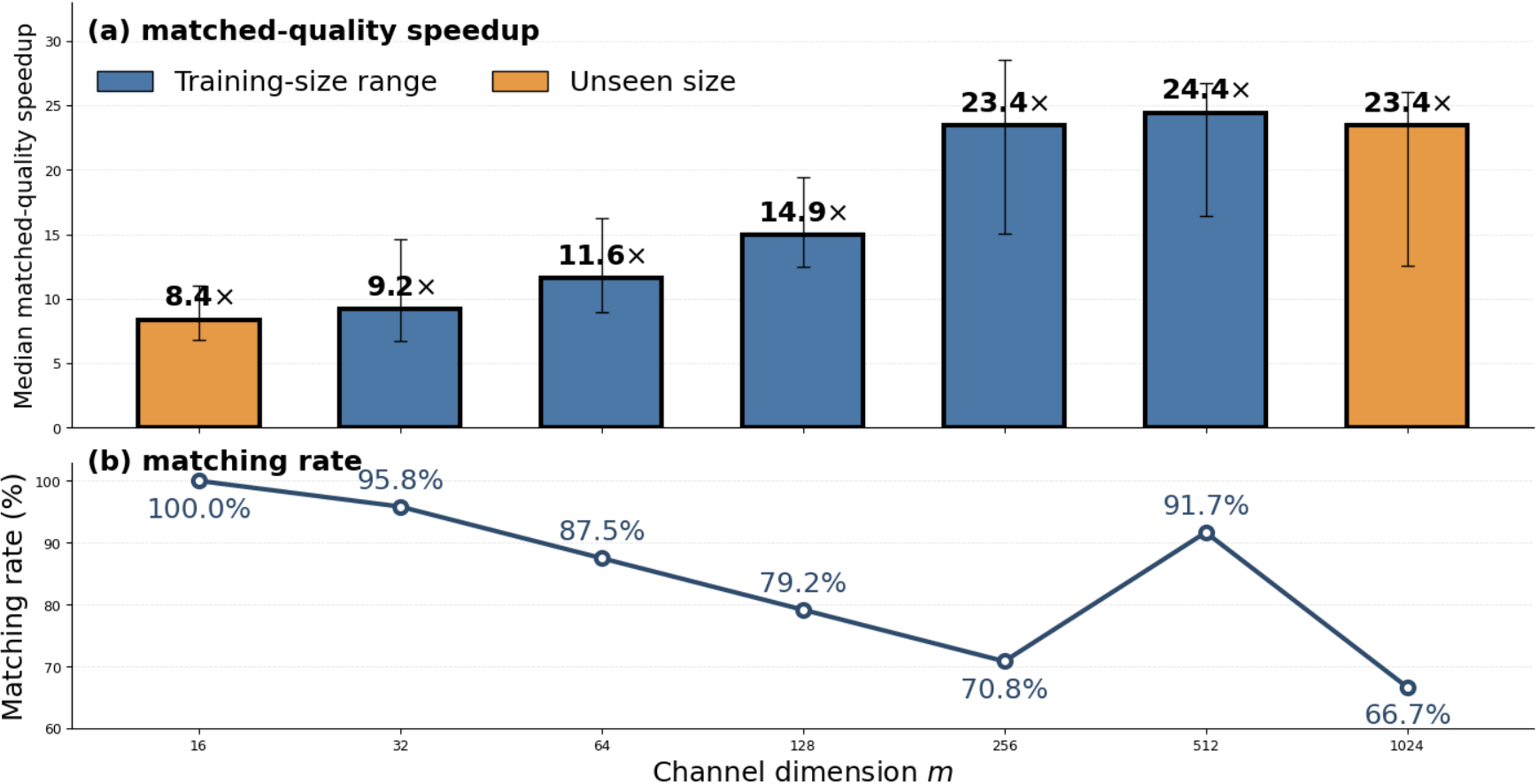}
\caption{
Matched-quality runtime comparison between the pretrained AWF
model and Mirror-Prox across channel dimensions. Panel (a)
shows the median speedup with interquartile-range error bars,
and Panel (b) shows the corresponding matching rate.
}
    \label{fig:matched_size_runtime}
\end{figure}

\color{black}
\begin{figure}
    \centering
    \color{black}
    \includegraphics[width=0.95\linewidth]{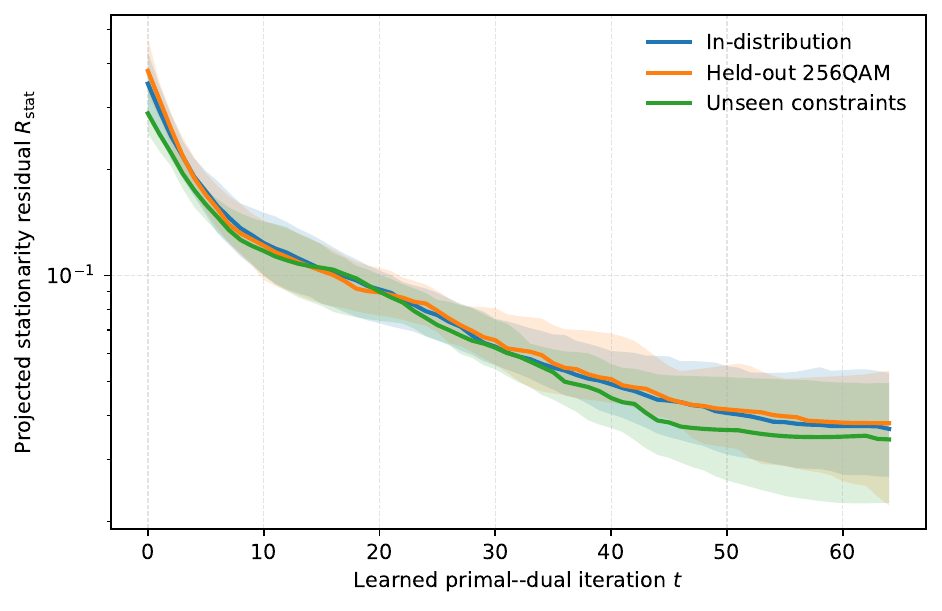}
    \caption{
Evolution of the projected-stationarity residual over the
learned primal--dual rollout. Curves and shaded regions denote
the median and interquartile range, respectively, across
representative in-distribution, held-out-modulation, and
unseen-constraint test instances.
}
    \label{fig:rollout_stationarity}
\end{figure}

Fig.~\ref{fig:rollout_stationarity} provides a diagnostic view
of the learned primal--dual dynamics. Across representative
in-distribution, held-out 256QAM, and unseen-constraint
settings, the median projected-stationarity residual is reduced
by approximately one order of magnitude over the fixed-depth
rollout. The similar trajectories across the three settings
indicate that the central stationarity-improving behavior of
the learned updates persists beyond the training distribution.

\color{black}
\color{black}
\subsection{Experiment II: Modulation Generalization}

\begin{table}[t]
\centering
\color{black}
\caption{
Generalization across modulation formats at $m=1024$.
}
\label{tab:modulation_generalization}
\setlength{\tabcolsep}{3.2pt}
\renewcommand{\arraystretch}{0.88}
\begin{tabular}{l l c c c}
\toprule
Modulation & Method & $J$ & $V$ & $R_{\rm stat}$ \\
\midrule
16QAM
& Direct-GNN
& 0.5660 & $6.23{\times}10^{-3}$ & 0.3128 \\
& AWF Model
& 0.5880 & $1.55{\times}10^{-3}$ & 0.0554 \\
& Mirror-Prox
& 0.5884 & $1.62{\times}10^{-3}$ & 0.0494 \\
\midrule

64QAM
& Direct-GNN
& 0.4763 & $1.66{\times}10^{-3}$ & 0.3381 \\
& AWF Model
& 0.5106 & $1.55{\times}10^{-3}$ & 0.0429 \\
& Mirror-Prox
& 0.5109 & $1.59{\times}10^{-3}$ & 0.0431 \\
\midrule

Mixed
& Direct-GNN
& 0.5511 & $1.17{\times}10^{-3}$ & 0.2945 \\
& AWF Model
& 0.5747 & $9.58{\times}10^{-4}$ & 0.0352 \\
& Mirror-Prox
& 0.5767 & $9.25{\times}10^{-4}$ & 0.0494 \\
\midrule

256QAM
& Direct-GNN
& 0.4380 & $4.20{\times}10^{-3}$ & 0.3524 \\
& AWF Model
& 0.4640 & $1.56{\times}10^{-3}$ & 0.0440 \\
& Mirror-Prox
& 0.4660 & $1.60{\times}10^{-3}$ & 0.0516 \\
\bottomrule
\end{tabular}
\end{table}

We next evaluate generalization across modulation formats.
Discrete constellations induce different mutual-information and
MMSE characteristics and hence different AWF dynamics. At evaluation, we consider 16QAM, 64QAM, mixed
16QAM/64QAM, and held-out 256QAM at the unseen large-scale
dimension $m=1024$. The same pretrained AWF foundation model
is used for all settings without retraining or
instance-specific fine-tuning. Direct-GNN is again included to
test whether a comparable one-shot neural representation can
provide the same modulation generalization without the learned
primal--dual rollout.

As shown in Table~\ref{tab:modulation_generalization}, the AWF
model attains objective values close to Mirror-Prox while
maintaining low inequality violation and stationarity residuals
across all four modulation settings. Most notably, on held-out
256QAM, Direct-GNN gives $R_{\rm stat}=0.3524$, whereas the
AWF model achieves $0.0440$ and Mirror-Prox $0.0516$.
The AWF residuals for 16QAM, 64QAM, and mixed modulation also
remain in a similar first-order quality regime to Mirror-Prox.
These results indicate transfer across modulation-dependent
mutual-information characteristics rather than specialization
to the modulation formats used during training.
\color{black}

\color{black}
\subsection{Experiment III: Constraint Generalization}

\begin{table}[t]
\centering
\color{black}
\caption{
Generalization across constraint structures.}
\label{tab:constraint_generalization}
\setlength{\tabcolsep}{3.2pt}
\renewcommand{\arraystretch}{0.88}
\begin{tabular}{l l c c c}
\toprule
Constraint & Method & $J$ & $V$ & $R_{\rm stat}$ \\
\midrule

Sparse
& Direct-GNN
& 0.4959 & $3.18{\times}10^{-3}$ & 0.3347 \\
& AWF Model
& 0.5262 & $2.37{\times}10^{-3}$ & 0.0386 \\
& Mirror-Prox
& 0.5277 & $2.16{\times}10^{-3}$ & 0.0472 \\
\midrule

Group
& Direct-GNN
& 0.4324 & $5.43{\times}10^{-3}$ & 0.3092 \\
& AWF Model
& 0.4355 & $1.38{\times}10^{-3}$ & 0.0446 \\
& Mirror-Prox
& 0.4391 & $1.57{\times}10^{-3}$ & 0.0220 \\
\midrule

Prefix
& Direct-GNN
& 0.5717 & $7.85{\times}10^{-4}$ & 0.2847 \\
& AWF Model
& 0.5978 & $1.41{\times}10^{-4}$ & 0.0375 \\
& Mirror-Prox
& 0.5976 & $3.25{\times}10^{-4}$ & 0.0315 \\
\midrule

Dense
& Direct-GNN
& 0.5026 & $4.99{\times}10^{-6}$ & 0.3078 \\
& AWF Model
& 0.5276 & $1.40{\times}10^{-5}$ & 0.0307 \\
& Mirror-Prox
& 0.5276 & $1.50{\times}10^{-5}$ & 0.0307 \\
\bottomrule
\end{tabular}
\end{table}

We further evaluate whether a single pretrained AWF foundation
model transfers across different linear-constraint structures.
Training uses random sparse nonnegative constraint matrices,
while evaluation additionally considers unseen group, prefix,
and dense constraints. All tests use $m=1024$, and the same
pretrained model is kept frozen. Direct-GNN is again included to
separate the effect of neural representation from that of the
learned primal--dual dynamics. Table~\ref{tab:constraint_generalization} shows that the AWF
model maintains low $R_{\rm stat}$ across all four constraint
families, ranging from $0.0307$ to $0.0446$, compared with
$0.2847$--$0.3347$ for Direct-GNN, while attaining objective
values close to Mirror-Prox. These results support reuse of the
same frozen AWF dynamics under constraint topologies not seen
during training.

\color{black}
Across Experiments I--III, Mirror-Prox matches the AWF
stationarity--feasibility target within $K_{\max}=3000$
iterations in 332 of 392 cases ($84.7\%$). Over these matched
cases, AWF achieves a median speedup of $12.76\times$, with an
interquartile range (IQR) of
$8.87\times$--$21.14\times$; unmatched cases are excluded
from the speedup statistics.

\color{black}
\subsection{Experiment IV: Generalization to Weighted Objectives}

We further test whether the learned AWF dynamics can be reused
when the channel-wise utility objective is reweighted. Starting
from the mercury/water-filling max--min problem in (\ref{eq:spatial-awf-mercury}), we
replace the original objective by $\max_{\mathbf p\geq 0,\,
\mathbf 1^{\mathsf T}\mathbf p=P,\,
\mathbf A\mathbf p\leq\widehat{\mathbf p}}
\min_{\substack{\mathbf n\geq 0,\,
\mathbf 1^{\mathsf T}\mathbf n=N}}
J_{\mathbf w}(\mathbf p,\mathbf n)$, $
J_{\mathbf w}(\mathbf p,\mathbf n)
= \frac{1}{m}\sum_{i=1}^{m} w_i I_i(\gamma_i)$,
where the positive weights $\mathbf w$ are unseen during pretraining. Thus, the feasible sets and max--min structure of (\ref{eq:spatial-awf-mercury}) remain unchanged, while the channel-wise utilities are
reweighted. We test three held-out weighting patterns over
$m\in\{32,64,128,256,512,1024\}$ with 10 independent instances
per size, yielding 180 weighted settings. The projected-stationarity residual is evaluated using the
corresponding weighted objective. As shown in
Table~\ref{tab:weighted_objective_generalization}, moderate
weighting shifts retain low residuals for most instances, whereas
the wider $\{0.25,1.75\}$ shift exhibits a heavier tail. Overall, the results indicate that the pretrained model maintains favorable first-order stationarity across a range of unseen weighted objectives, although stronger weighting shifts lead to a heavier residual tail. These results provide further evidence of zero-shot generalization across weighted objective variants.

\begin{table}[t]
\centering
\color{black}
\caption{Zero-shot reuse under unseen weighted AWF objectives.}
\label{tab:weighted_objective_generalization}
\setlength{\tabcolsep}{3.2pt}
\renewcommand{\arraystretch}{0.94}
\begin{tabular}{l l c c c}
\toprule
Pattern & Weights
& Mean $R_{\rm stat}$
& Med./P90
& $R_{\rm stat}\!\leq\!0.05$ \\
\midrule
Random positive
& $\mathcal{U}[0.5,1.5]$
& 0.0274 & 0.0208/0.0489 & 54/60 \\
Grouped
& $\{0.5,1.5\}$
& 0.0355 & 0.0294/0.0648 & 50/60 \\
Wide two-level
& $\{0.25,1.75\}$
& 0.0521 & 0.0458/0.0863 & 35/60 \\
\bottomrule
\end{tabular}
\end{table}
\color{black}

\color{black}
\subsection{Learning Baseline and Architectural Ablations}
\label{subsec:ablation}

We conduct controlled experiments to identify the
contributions of the learned primal--dual rollout,
constraint-aware graph representation, and
distribution-dependent conditioning. All variants are trained
from scratch using the same data generator, training
distribution, objective, and training protocol, and are
evaluated on the same predefined instances.

Table~\ref{tab:ablation} shows that Direct-GNN, which removes
the learned primal--dual rollout and predicts the allocation in
a single forward pass, produces the largest degradation. Its
projected-stationarity residual is approximately $6.8\times$,
$7.3\times$, and $7.7\times$ larger than that of the full AWF
model for the in-distribution, held-out 256QAM, and
unseen-constraint evaluations, respectively, while its average
inequality violation is about $6.2\times$ larger. Thus, a
comparable neural representation alone does not provide the
same first-order quality as the learned dynamics. Removing the constraint-aware GNN mainly affects feasibility:
the average inequality violation increases from
$8.84\times10^{-4}$ to $1.71\times10^{-3}$, while
stationarity remains comparable. Removing
distribution-dependent conditioning leads to a smaller but
consistent degradation, including for held-out 256QAM.
Together, these ablations indicate complementary roles of the
learned update dynamics, constraint representation, and
modulation-dependent conditioning.

\begin{table}[t]
\centering
\color{black}
\caption{
Learning baseline and component ablations.
}
\label{tab:ablation}
\setlength{\tabcolsep}{2pt}
\renewcommand{\arraystretch}{0.88}
\begin{tabular}{l c c c c c}
\toprule
Method &
In-dist. &
256QAM &
Unseen const. &
$J$ &
$V$ \\
& $R_{\rm stat}$ &
$R_{\rm stat}$ &
$R_{\rm stat}$ & & \\
\midrule
Direct-GNN
& 0.3144 & 0.3480 & 0.3037
& 0.5058 & $5.44{\times}10^{-3}$ \\

w/o Constraint GNN
& 0.0434 & 0.0450 & 0.0356
& 0.5217 & $1.71{\times}10^{-3}$ \\

w/o Dist. Conditioning
& 0.0471 & 0.0490 & 0.0406
& 0.5214 & $9.11{\times}10^{-4}$ \\

Full AWF
& 0.0461 & 0.0478 & 0.0395
& 0.5212 & $8.84{\times}10^{-4}$ \\
\bottomrule
\end{tabular}
\end{table}

\noindent\textit{Graph architecture choice.}
To examine the role of explicit bipartite message passing, we
replace the constraint GNN with a masked bipartite Graph
Transformer while keeping the remaining AWF architecture and
learned primal--dual dynamics unchanged. The two variants have
comparable capacity and use the
same training and evaluation instances.

Across all generalization instances in Table \ref{tab:graph_architecture}, the two architectures
achieve nearly identical objective and stationarity quality:
the bipartite GNN gives $J=0.5212$ and
$R_{\rm stat}=0.0440$, compared with $J=0.5214$ and
$R_{\rm stat}=0.0445$ for the Graph Transformer. The main
difference is feasibility, with average inequality violation
increasing from $8.84\times10^{-4}$ to
$1.43\times10^{-3}$ under masked attention. Their scaling behavior also differs: latency is similar for
$m\leq64$, while at $m=1024$ the median latency is 242.2 ms
for the bipartite GNN and 271.2 ms for the Graph Transformer.
We therefore retain bipartite message passing for the sparse
constraint graphs considered here due to its favorable
feasibility--scalability balance, while both representations
remain viable.
\begin{table}[t]
\centering
\color{black}
\caption{
Graph representation comparison, with runtime reported at $m=1024$.
}
\label{tab:graph_architecture}
\setlength{\tabcolsep}{3.2pt}
\renewcommand{\arraystretch}{0.95}
\begin{tabular}{l c c c c}
\toprule
Graph module &
$J$ &
$V$ &
$R_{\rm stat}$ &
Time$_{1024}$ (ms) \\
\midrule
Bipartite GNN
& 0.5212
& $8.84{\times}10^{-4}$
& 0.0440
& 242.2 \\

Graph Transformer
& 0.5214
& $1.43{\times}10^{-3}$
& 0.0445
& 271.2 \\
\bottomrule
\end{tabular}
\end{table}

\color{black}

\color{black}
\subsection{NTN-Inspired Channel and Interference Sensitivity}
\label{subsec:ntn_sensitivity}

Together with the constraint-generalization experiment, we
evaluate the frozen AWF foundation model under NTN-inspired
large-scale channel variations and adversarial interference
budgets. This experiment is intended as a controlled
snapshot-level stress test rather than a complete standardized
link-level or coexistence simulation.

Specifically, for $m=64$, each base channel realization is
perturbed using a snapshot-level large-scale factor motivated by
the S-band LOS shadow-fading setting in 3GPP TR~38.811~\cite{3GPP38811}.
For each base snapshot, we independently sample
$X_{\rm SF}\sim\mathcal{N}(0,0.72^2)\ {\rm dB}$,
and define $s_{\rm NTN}=10^{-X_{\rm SF}/10}$, $
\widetilde{\boldsymbol{\beta}}
=
s_{\rm NTN}\boldsymbol{\beta}.$
The same factor $s_{\rm NTN}$ is applied to all channel-gain
entries within a given snapshot, and is kept fixed when sweeping
the interference budget $N/m$. Across the 20 frozen base
snapshots, the realized values of $s_{\rm NTN}$ range from
$0.665$ to $1.370$, with a median of $1.013$.
We independently vary the aggregate adversarial interference
budget according to $\frac{N}{m}\in\{0.05,\,0.15,\,0.30,\,0.50\}$.
For each base instance,
$\widetilde{\boldsymbol{\beta}}$, $\boldsymbol{\sigma}$, $P$,
$\mathbf A$, and $\hat{\mathbf p}$ are fixed across operating
points, while only $N$ is varied. The learned model remains
completely frozen, and Mirror-Prox uses the same configuration
and projected-stationarity metric as in the main experiments.

\begin{table}[t]
\centering
\color{black}
\caption{
NTN-inspired channel and interference sensitivity over
20 paired base instances.
}
\label{tab:ntn_sensitivity}
\resizebox{\columnwidth}{!}{
\begin{tabular}{c|ccc|cccc}
\hline
& \multicolumn{3}{c|}{Learned AWF}
& \multicolumn{4}{c}{Mirror-Prox} \\
$N/m$
& $J$
& $V$
& $R_{\rm stat}$
& $J$
& $V$
& $R_{\rm stat}$
& Match \\
\hline
0.05
& 0.6743
& $4.12{\times}10^{-4}$
& 0.1423
& 0.6759
& $9.96{\times}10^{-4}$
& 0.0932
& 90\% \\

0.15
& 0.5696
& $3.09{\times}10^{-4}$
& 0.0561
& 0.5729
& $8.65{\times}10^{-4}$
& 0.0436
& 90\% \\

0.30
& 0.4718
& $2.50{\times}10^{-4}$
& 0.0361
& 0.4749
& $9.32{\times}10^{-4}$
& 0.0236
& 100\% \\

0.50
& 0.3879
& $2.87{\times}10^{-4}$
& 0.0326
& 0.3915
& $7.98{\times}10^{-4}$
& 0.0233
& 100\% \\
\hline
\end{tabular}}
\end{table}

\color{black}
Table~\ref{tab:ntn_sensitivity} shows a physically consistent
sensitivity trend. As the adversarial interference budget
increases from $N/m=0.05$ to $0.50$, the objective decreases
for both methods, as expected from the max--min structure of
AWF, while $R_{\rm stat}$ of the AWF model decreases from
$0.1423$ to $0.0326$. The weakest-interference setting is the
most challenging for the fixed-depth learned dynamics, with the
interference-side residual being dominant. Throughout the sweep,
the AWF model maintains an average inequality violation below
$4.2\times10^{-4}$, while Mirror-Prox reaches the corresponding
stationarity--feasibility target on $90\%$, $90\%$, $100\%$,
and $100\%$ of the instances, respectively. Overall, the AWF model retains low feasibility error and predictable
snapshot-level sensitivity under the considered NTN-inspired
channel and interference variations without NTN-specific
retraining.
\color{black}

\section{Conclusion}

\color{black}
This paper has introduced AWF as a unified framework for competitive spectrum and spatial resource allocation. Gaussian AWF admits a convex--concave formulation with a unique saddle point, whereas practical discrete constellations lead to generally nonconvex mercury/water-filling formulations. To address the latter, we
have developed a domain-specific wireless foundation model that learns AWF dynamics across varying channel dimensions, constraint structures, and
modulation distributions. Theoretical analysis has established projected-stationarity/KKT consistency and conditional local convergence under explicit
regularity and stability conditions. Experiments have demonstrated generalization across problem sizes, modulation formats, constraint structures,
and held-out channel-weighted objectives, together with snapshot-level sensitivity under NTN-inspired channel and interference variations. These results show that the same pretrained model can be reused across these variations without instance-specific retraining. Future work of interest includes broader transfer and adaptation across distinct wireless resource-allocation task families, multi-agent extensions, and temporally coupled settings
involving time-varying channels and network-level coordination in large-scale LEO systems.

\color{black}
% This paper has proposed the AWF as a unified framework for competitive spectrum and spatial resource allocation. 
% While Gaussian adversarial water-filling gives rise to a strongly convex--concave game, practical discrete constellations lead to generally nonconvex mercury/water-filling formulations. 
% To address this challenge, we have proposed a wireless foundation model that learns water level search dynamics across varying channel dimensions, constraints, and modulation distributions. 
% The architecture integrates permutation-invariant channel representations, constraint-aware message passing, and global latent variables that reflect the optimal adversarial water level. 
% Theoretical analysis characterized conditional KKT consistency and local convergence under regularity and contraction conditions. Experiments demonstrated empirical generalization across problem sizes, constraints, and modulation formats, while achieving significant runtime improvements over iterative baselines. Interesting problems for future work include extending this model to multi-agent wireless optimization problems, as well as incorporating additional system dynamics such as time-varying channels and network-level coordination in large-scale LEO deployments.

\bibliographystyle{IEEEtran}
\bibliography{ref}

\vfill

\end{document}